\documentclass[11pt]{article}
\usepackage{amsmath}
\usepackage{graphicx}
\usepackage{enumerate}
\usepackage{natbib}
\usepackage{url}
\usepackage{amsfonts,amssymb}
\usepackage{bm}
\usepackage{amsthm}
\usepackage{float}
\usepackage{setspace}
\usepackage{caption}
\usepackage{algpseudocode}
\usepackage{algorithm}
\usepackage{algorithmicx}
\usepackage{booktabs}
\usepackage{longtable}
\usepackage{array} 
\usepackage{multirow}
\usepackage{siunitx}
\usepackage{makecell}

\usepackage[colorlinks]{hyperref}
\hypersetup{
    colorlinks,
    linkcolor={blue},
    filecolor={blue}, 
    urlcolor={blue},
    citecolor={blue}
}
\usepackage[capitalise,noabbrev]{cleveref}

\usepackage{tikz}           
\usetikzlibrary{
    fit,
    decorations.pathmorphing, 
    positioning,              
    arrows,                  
    decorations.pathreplacing,
    calc                
}

\usepackage{natbib}

\newtheorem{thm}{Theorem}
\newtheorem{lemma}{Lemma} 
\newtheorem{definition}{Definition}

\newtheorem{assumption}{Assumption}
\newtheorem{condition}{Condition}
\newtheorem{corollary}{Corollary}

\newcommand{\tR}{\mathbb{R}}

\newcommand{\cP}{\mathcal{P}}

\newcommand{\cN}{\mathcal{N}}

\newcommand{\cS}{\mathcal{S}}

\newcommand{\bS}{\mathbf{S}}

\newcommand{\btau}{\bm{\tau}}

\newcommand{\bSigma}{\bm{\Sigma}}

\newcommand{\bE}{\mathbf{E}}

\newcommand{\bA}{\mathbf{A}}

\newcommand{\bO}{\mathbf{O}}

\newcommand{\cO}{\mathcal{O}}

\newcommand{\bU}{\mathbf{U}}

\newcommand{\bD}{\mathbf{D}}

\newcommand{\bX}{\mathbf{X}}

\newcommand{\bY}{\mathbf{Y}}

\newcommand{\bw}{\mathbf{w}}
\newcommand{\bW}{\mathbf{W}}
\newcommand{\bH}{\mathbf{H}}

\newcommand{\bz}{\bm{z}}

\newcommand{\bzero}{\bm{0}}
\newcommand{\bone}{\bm{1}}

\newcommand{\cA}{\mathcal{A}}

\usepackage{xcolor}

\newcommand{\redcol}{\textcolor{red}}

\newcommand{\tE}{\mathbb{E}}
\newcommand{\tP}{\mathbb{P}}

\newcommand{\hbY}{\hat{\mathbf{Y}}}
\newcommand{\hY}{\hat{Y}}

\newtheorem{example}{Example}

\begin{document}

\def\spacingset#1{\renewcommand{\baselinestretch}%
{#1}\small\normalsize} \spacingset{1}

\title{\bf Causal Inference for Preprocessed Outcomes with an Application to Functional Connectivity}
\author{Zihang Wang, Razieh Nabi, Benjamin B. Risk\\
    Department of Biostatistics and Bioinformatics, Emory University}
\maketitle

\bigskip

\begin{abstract}
In biomedical research, repeated measurements within each subject are often processed to remove artifacts and unwanted sources of variation. The resulting data are used to construct derived outcomes that act as proxies for scientific outcomes that are not directly observable. Although intra-subject processing is widely used, its impact on inter-subject statistical inference has not been systematically studied, and a principled framework for causal analysis in this setting is lacking. In this article, we propose a semiparametric framework for causal inference with derived outcomes obtained after intra-subject processing. This framework applies to settings with a modular structure, where intra-subject analyses are conducted independently across subjects and are followed by inter-subject analyses based on parameters from the intra-subject stage. We develop multiply robust estimators of causal parameters under rate conditions on both intra-subject and inter-subject models, which allows the use of flexible machine learning. We specialize the framework to a mediation setting and focus on the natural direct effect. For high dimensional inference, we employ a step-down procedure that controls the exceedance rate of the false discovery proportion. Simulation studies demonstrate the superior performance of the proposed approach. We apply our method to estimate the impact of stimulant medication on brain connectivity in children with autism spectrum disorder.

\end{abstract}

\noindent%
\vfill

\newpage
\spacingset{1.9}

\section{Introduction}
\label{sec:intro}

Causal inference provides a statistical foundation for studying the impact of treatments, exposures, or biological conditions on individuals in observational biomedical research \citep{imbens2015causal}. Advancements in data collection technologies now allow the collection of many observations within each subject. For example, in functional magnetic resonance imaging (fMRI), hundreds images of the blood oxygen level–dependent (BOLD) signal are collected over time for each participant. These intra-subject measurements contain information about subject-level outcomes of interest. In practice, intra-subject processing is typically necessary prior to inter-subject analysis. Such processing aims to remove artifacts and unwanted sources of variation. Common examples include nuisance regression of motion alignment parameters in fMRI analysis, regression of eye blink and other artifacts in electroencephalography analysis, and neuropil or global signal regression in calcium imaging analysis \citep{power2014methods,di2016new,stringer2019computational}. Existing intra-subject processing pipelines vary widely across studies. However, theoretical guarantees when using intra-subject processing are lacking. Additionally, modern machine learning can improve intra-subject processing \citep{chaudhary2022fast,zhang2024motion}. The impact of adaptive learning in this setting remains largely unexplored.

Our motivating example is to study stimulant medication effects on brain functional connectivity (FC) in school-aged children. Resting-state fMRI (rs-fMRI) data are collected while children are not performing any specific tasks. FC is defined as statistical dependencies between brain regions, and it is commonly measured using Pearson correlation \citep{power2014methods,ciric2017benchmarking}. The intra-subject measurements include rs-fMRI signals and time courses of head movement alignment parameters. Intra-subject processing involves linear regression of movement alignment parameters on rs-fMRI data within each subject to remove motion artifacts. Due to non-linear and complex motion artifacts, scrubbing is widely used in which volumes with excessive motion are removed \citep{cox1996afni,power2014methods,ciric2017benchmarking}. In clinical cohorts, this process can remove more than 50\% of intra-subject data \citep{yan2013comprehensive}. In addition, subjects with limited usable data are often excluded. Subject removal may introduce selection bias because children who move more tend to have more behavioral problems \citep{cosgrove2022limits,nebel2022accounting}. A recent large study of brain-behavior associations in children excluded more than 60\% of participants \citep{marek2022reproducible}. 
Modern machine learning methods for denoising neuroimaging data can effectively model relationships among intra-subject measurements and may help retain more usable data \citep{chaudhary2022fast,manzano2024denoising}. However, their application to nuisance regression in rs-fMRI remains limited. A framework for valid causal inference with conditions on intra-subject processing is still lacking.

Robust estimators developed within semiparametric inference theory offer a promising solution to allow flexible machine learning models to be used in causal inference \citep{robins1994estimation,scharfstein1999adjusting,van2011targeted}. For example, the doubly robust estimators for the average treatment effect (ATE) are consistent if either the propensity score model or the outcome regression model is consistently estimated \citep{cao2009improving}. Under additional regularity conditions ensuring that second-order remainder terms are asymptotically negligible, these estimators admit an asymptotically linear expansion with the efficient influence function and attain the semiparametric efficiency bound. In particular, this holds when products of propensity score and outcome regression estimation errors converge at a rate faster than $n^{-1/2}$, allowing both nuisance components to be estimated flexibly at rates slower than $n^{-1/2}$ individually. This property substantially reduces sensitivity to model misspecification compared with fully parametric approaches \citep{laan2003unified,tsiatis2006semiparametric}. However, despite their usefulness for efficient causal modeling with flexible nuisance estimation, existing methods do not directly apply to the setting studied here, where outcomes are estimated from intra-subject processing. In this case, additional estimation error is introduced prior to the inter-subject analysis, and standard doubly robust theory does not account for the impact of this intra-subject processing on asymptotic linearity or efficiency. This motivates the development of a framework that allows data-driven intra-subject processing while preserving robust inter-subject estimation within a semiparametric causal inference paradigm.

The unobservable nature of the outcomes places the problem within the derived outcome framework. \cite{qiu2023unveiling} proposed the derived outcome framework that explicitly allows outcomes to be unobservable. Identification and estimation of ATE are established by using an outcome estimator as a bridge, under a modified ignorability condition and a regularity condition on the accuracy of outcome estimates. Their inverse probability weighting (IPW) estimators rely on correctly specified propensity score models and can be sensitive to model misspecification. To improve robustness, \citet{du2025causal} introduced doubly robust estimators into the derived outcome framework and specialized their approach to standardized average treatment effects and quantile treatment effects in genomic studies. However, these frameworks assume that variability in the derived outcome arises only from outcome estimation based on intra-subject observations. In practice, these observations are often residuals obtained after intra-subject processing. This additional processing step introduces another source of uncertainty, since the relationship among the raw intra-subject measurements must also be estimated.

Motivated by the limitations described above, we propose a semiparametric inference framework under a hierarchical setting for causal analysis with derived outcomes constructed from intra-subject processed measurements. As far as we know, this is the first work to study causal inference after intra-subject processing within a theoretical framework. We define causal estimands under a mediation structure, which is biologically meaningful in our motivating rs-fMRI study. We develop multiply robust and efficient estimators for the proposed estimands by integrating semiparametric theory for causal mediation analysis into our intra-subject processed derived outcome framework \citep{tchetgen2012semiparametric,vanderweele2015explanation}. Both intra-subject processing functions and inter-subject models are allowed to belong to broad semiparametric classes, which supports flexible estimation and reduces sensitivity to model misspecification. We establish asymptotic normality for individual estimators and adopt a step-down multiple testing procedure from \cite{qiu2023unveiling} and \cite{du2025causal}. An application to resting-state fMRI data from children with autism spectrum disorder (ASD) demonstrates the practical utility of the proposed framework. We use flexible machine learning algorithms to remove motion artifacts, which removes the need for data scrubbing and participant removal. We then estimate the direct effect of stimulant medication on brain FC.

The remainder of the paper is organized as follows. Section~\ref{sec:causal_intra} presents the hierarchical model incorporating both intra-subject and inter-subject variables and states the assumptions required for identification of the causal estimands. Section~\ref{sec:multi_robust} introduces a multiply robust estimator and examines its statistical properties. Section~\ref{sec:simultaneous} establishes asymptotic normality of the estimator and describes a multiple testing procedure for controlling the false discovery proportion. Section~\ref{sec:simulations} presents simulation studies comparing the proposed approach with existing methods. Section~\ref{sec:application} applies the framework to rs-fMRI data from children with autism spectrum disorder to assess stimulant medication effects on FC. Section~\ref{sec:discussion} discusses limitations and directions for future research. Details of derivations and technical proofs are provided in the Web Supplement Section A.

\section{Causal Inference with Intra-subject Processing}
\label{sec:causal_intra}

We present a general framework for causal analysis with derived outcomes that account for intra-subject processing. The framework is developed in the context of causal mediation analysis (Panel A in \cref{fig:brain_diagrams}) and reduces to the ATE in the absence of a mediator. We then illustrate the framework with a specific application to functional connectivity studies (Panel B in \cref{fig:brain_diagrams}). Let $\bone_p \in \tR^p$ denote a vector with all components equal to $1$, and let $\bzero_p \in \tR^p$ denote the zero vector. Let $\bO_n = (O_1, \ldots, O_n)$ be an independent and identically distributed sample drawn from $\tP \in \cP$, where $\cP$ is the statistical model. For a measurable function $f$, we write $\tP [f(O)] = \tE [f(O)] = \int f(o) d \tP(o)$, and let $\hat{\tP}$ denote an estimator of $\tP$. The empirical distribution is denoted by $\tP_n$, so that $\tP_n [f(O)] = \int f(o) d \tP_n(o) = \frac{1}{n} \sum_{i=1}^n f(O_i)$. We define the $L_p$ norm of $f$ as $\| f \|_{p} = \{ \int f^p(o) d \tP(o) \}^{1/p}$.

\subsection{The Hierarchical Model with Intra-subject Processing}
\label{sec:hie_model}

\begin{figure}[t]
\scalebox{0.4}{
    \begin{tikzpicture}[>=stealth, node distance=1cm]
        \tikzstyle{format} = [thick, circle, minimum size=3.0mm, inner sep=2pt]
        \tikzstyle{square} = [draw, thick, minimum size=4.5mm, inner sep=2pt]

        \begin{scope}[xshift=0cm, yshift=0cm]
            \path[->, thick, line width=0.07cm]
            
            node[] (a) {\scalebox{1.5}{$\underset{\text{\large (Treatment)}}{A}$}}
            node[above of=a, xshift=-1.5cm, yshift=6cm] (Alabel) {\huge \textbf{Panel A}}
            
            node[right of=a, xshift=6cm] (ui) {\scalebox{1.5}{$\underset{ \substack{\text{\large \redcol{(Indirect}} \\ \text{\large \redcol{Latent Factor)}}}}{\redcol{\bU_I}}$}}
            
            node[right of=ui, xshift=7cm] (m) {\scalebox{1.5}{$\underset{\text{\large (Mediator)}}{M}$}}
            
            node[right of=m, xshift=8cm, yshift=0cm] (xtht) {\scalebox{1.5}{$\underset{\text{\large (Response, Nuisance)}}{\bX_T, \bH_T}$}}

            node[above of=ui, xshift=0cm, yshift=4cm] (w) {\scalebox{1.5}{$\underset{\text{\large (Confounders)}}{\bW}$}}
            
            node[above of=m, xshift=0cm, yshift=4cm] (ud) {\scalebox{1.5}{$\underset{\substack{\text{\large \redcol{(Direct}} \\ \text{\large \redcol{Latent Factor)}}}}{\redcol{\bU_D}}$}}
            
            node[right of=ud, xshift=8cm, yshift=0cm] (x1h1) {\scalebox{1.5}{$\underset{\text{\large (Response, Nuisance)}}{\bX_1, \bH_1}$}}
            
            node[right of=xtht, xshift=6cm, yshift=2.5cm] (y) {\scalebox{1.5}{$\underset{\text{\large (Outcome)}}{\bY}$}}

            node[above of=xtht, xshift=0cm, yshift=2cm] (dots) {\scalebox{3}{$\vdots$}}
            
            (w) edge[black, bend left=0] (a) 
            (w) edge[black, bend left=0] (ui)
            (w) edge[black, bend left=0] (ud)
            (w) edge[black, bend right=0] (m)
            (w) edge[black, bend left=20] (x1h1)
            (w) edge[black] (xtht)
            
            (a) edge[black] (ui) 
            (a) edge[black, bend left=45] (ud) 
            
            (ui) edge[black] (m) 
            (m) edge[black] (x1h1)
            (m) edge[black] (xtht)
            
            (ud) edge[black] (x1h1) 
            (ud) edge[black] (xtht) 

            (x1h1) edge[black] (y) 
            (xtht) edge[black] (y) 
             
            ;
            
        \end{scope}

        \begin{scope}[xshift=0cm, yshift=-10.5cm]
            \path[->, thick, line width=0.07cm]
            
            node[] (a) {\scalebox{1.5}{$\underset{\text{\large (Treatment)}}{A}$}}
            node[above of=a, xshift=-1.5cm, yshift=6cm] (Alabel) {\huge \textbf{Panel B}} 
            
            node[right of=a, xshift=5cm] (um) {\scalebox{1.5}{$\underset{ \substack{\text{\large \redcol{(Motion-related}} \\ \text{\large \redcol{Brain Wiring)}}}}{\redcol{\bU_M}}$}}
            
            node[right of=um, xshift=6cm] (sm) {\scalebox{1.5}{$\underset{ \substack{\text{\large \redcol{(Motion}} \\ \text{\large \redcol{Neural activity)}}}}{ \redcol{\bS_{M}}}$}}
            
            node[above of=sm, xshift=0cm, yshift=4.5cm] (ub) {\scalebox{1.5}{$\underset{\substack{\text{\large \redcol{(Non-motion}} \\ \text{\large \redcol{Brain Wiring)}}}}{ \redcol{\bU_B}}$}}

            node[right of=ub, xshift=5cm, yshift=-2.5cm] (sb) {\scalebox{1.5}{$\underset{\substack{\text{\large \redcol{(Non-motion}} \\ \text{\large \redcol{Neural activity)}}}}{ \redcol{\bS_{B}}}$}}
            
            node[right of=ub, xshift=13cm, yshift=0cm] (x1h1) {\scalebox{1.5}{$\underset{\substack{\text{\large (fMRI Signal,} \\ \text{\large  Transient Movement)}}}{\bX_1, \bH_1}$}}

            node[right of=sm, xshift=13cm, yshift=0cm] (xtht) {\scalebox{1.5}{$\underset{\substack{\text{\large (fMRI Signal,} \\ \text{\large  Transient Movement)}}}{\bX_T, \bH_T}$}}
            
            node[right of=xtht, xshift=6cm, yshift=2.5cm] (y) {\scalebox{1.5}{$\underset{\substack{\text{\large (Functional} \\ \text{\large Connectivity)}}}{\bY}$}}
            
            node[above of=um, xshift=0cm, yshift=4.5cm] (w) {\scalebox{1.5}{$\underset{\text{\large (Confounders)}}{\bW}$}}
            
            node[right of=sm, xshift=5cm, yshift=0cm] (m) {\scalebox{1.5}{$\underset{\substack{\text{\large (Motion} \\ \text{\large Trait)}}}{M}$}}

            node[above of=xtht, xshift=0cm, yshift=2.3cm] (dots) {\scalebox{3}{$\vdots$}}
            
            (w) edge[black, bend left=0] (a) 
            (w) edge[black, bend left=0] (um)
            (w) edge[black, bend left=0] (ub)
            (w) edge[black, bend right=0] (sb)
            (w) edge[black, bend right=0] (sm)
            
            (a) edge[black] (um) 
            (a) edge[black, bend left=50] (ub) 
            
            (um) edge[black] (sm) 
            (ub) edge[black] (sb)
            
            (sm) edge[black] (m) 
            (m) edge[black] (x1h1) 
            (m) edge[black] (xtht) 

            (sb) edge[black] (x1h1) 
            (sb) edge[black] (xtht) 

            (x1h1) edge[black] (y) 
            (xtht) edge[black] (y) 
             
            ;
            
        \end{scope}
        
    \end{tikzpicture}
    }
\caption{\textbf{Panel A}: The general diagram of causal mediation analysis with derived outcomes in the presence of intra-subject nuisance. Treatment is denoted by $A \in \{0,1\}$. $\bU_D \in \tR^{l_D}$ and $\bU_I \in \tR^{l_I}$ are latent factors along the direct and indirect treatment–outcome pathways respectively. The observed mediator $M \in \tR$ is generated from the indirect latent factor $\bU_I$. Intra-subject repeated measurements of responses $\bX_1,\cdots,\bX_T \in \tR^V$ and nuisance $\bH_1,\cdots,\bH_T \in \tR^p$ are influenced by both $\bU_D$ and $M$. The vector of derived outcomes $\bY \in \tR^J$ is constructed from $\bX$ and $\bH$. \textbf{Panel B}: The diagram for causal inference in functional connectivity studies. $\bU_M \in \tR^{l_M}$ and $\bU_B \in \tR^{l_B}$ are latent factors representing motion-related and non-motion-related brain wiring, respectively. Non-motion neural activity $\bS_{B} \in \tR^{l_{sb}}$ is generated by $\bU_B$, while motion-related activity $\bS_{M} \in \tR^{l_{sm}}$ is generated by $\bU_M$. Both $\bS_{B}$ and $\bS_{M}$ are unobserved. $\bS_{M}$ further generates the observed motion trait $M$, often represented as mean framewise displacement. Both $\bS_{B}$ and $M$ influence the fMRI response $\bX_t$ and transient movement variables $\bH_t$ at each scan time $t=1,\cdots,T$. These are then used to construct the $J$-dimensional derived outcomes of functional connectivity $\bY$.
}
\label{fig:brain_diagrams}
\end{figure}
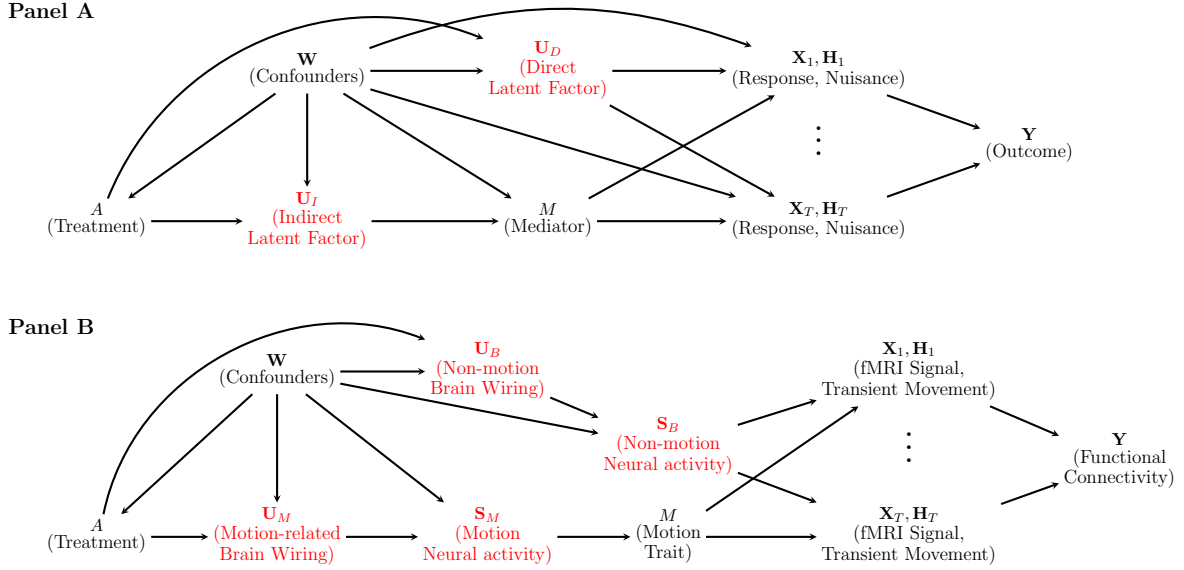

For the Panel A in \cref{fig:brain_diagrams}, consider a random sample of $n$ subjects with variables $\{ A, \bW, \bU_{I}(a), \bU_{D}(a), M(a), \bX(a,m), \bH(a,m) \}$. For each subject, let $A \in \{0,1\}$ denote the treatment indicator, where $A = 1$ indicates treatment receipt and $A = 0$ otherwise. Let $\bW \in \tR^q$ denote a vector of confounders that may jointly influence all variables in the diagram. We define $\bU_I(a) \in \tR^{l_I}$ and $\bU_D(a) \in \tR^{l_D}$ as latent factors corresponding to the indirect and direct causal pathways, respectively. The mediator $M(a) \in \tR$ is generated from the indirect latent factor $\bU_I(a)$ and together with the direct latent factor $\bU_D(a)$ determines the generating distribution of the intra-subject measurements, including the responses $\bX(a,m) = [\bX_1(a,m), \ldots, \bX_T(a,m)] \in \tR^{V \times T}$ and the nuisance variables $\bH(a,m) = [\bH_1(a,m), \ldots, \bH_T(a,m)] \in \tR^{p \times T}$, where $T$ is the number of repeated measurements on the subject. The vector of subject-level potential outcomes of interest is $\bY(a,m) \in \tR^J$, defined as a functional that maps the support of the responses $\bX_t(a,m)$ and nuisances $\bH_t(a,m)$, conditional on the latent factors $\bU_D(a)$, $\bU_I(a)$, and the confounders $\bW$, to a $J$-dimensional space, as shown in the following \cref{def:Y_def}.

\begin{definition}\label{def:Y_def}
(Outcome with intra-subject processing).
The subject-level potential outcome $\bY(a,m) \in \tR^J$ is defined through intra-subject processing of the random variables $\bX_t(a,m) \in \tR^V$ and $\bH_t(a,m) \in \tR^p$ as
\begin{equation}
\bY(a,m) = g\left[ \bX_t(a,m) - \textbf{f}(\bH_t(a,m)) \mid \bU_D(a), \bU_I(a), \bW \right],
\end{equation}
where $\textbf{f} = \{ f_v: \tR^p \to \tR;\ v = 1, \ldots, V \}$ denotes the collection of intra-subject processing functions, and $g: \tR^V \to \tR^J$ is a pre-specified outcome function mapping intra-subject processed residuals to subject-level potential outcomes.
\end{definition}

Panel B in \cref{fig:brain_diagrams} illustrates the diagram for a FC study, which represents the complex biological processes underlying brain activity. The treatment $A$ influences $\bU_B(a) \in \tR^{l_B}$, a latent factor that captures non-motion related neuronal wiring of the brain, which in turn determines the corresponding non-motion time-varying neural activity $\bS_B(a) \in \tR^{l_{sb}}$. Both $\bU_B(a)$ and $\bS_B(a)$ are latent factors. The treatment $A$ may also affect neuronal wiring related to head motion during scanning, represented by $\bU_M(a) \in \tR^{l_M}$, which subsequently governs motion related neural activity $\bS_M(a) \in \tR^{l_{sm}}$, and both $\bU_M(a)$ and $\bS_M(a)$ are also latent factors. The mediator $M(a)$ represents the subject’s motion trait and is commonly measured by mean framewise displacement (FD). In our data application, there are impacts of $M$ on the correlations between brain locations composing $\bX_t$ even after removing the estimated contributions of $f\{\bH_t(a,m)\}$, which motivates the decomposition of natural direct and indirect effects. The non-motion neural activity $\bS_B(a)$ and the motion level $M(a)$ jointly influence the potential responses $\bX(a,m) \in \tR^{V \times T}$, corresponding to intra-subject fMRI BOLD signals across $V$ brain regions and $T$ functional scans, as well as the potential nuisances $\bH(a,m) \in \tR^{p \times T}$, which contain $p$ movement alignment parameters across length $T$. In addition, subject-level confounders $\bW = (W_1, \dots, W_q)^\top \in \tR^q$, such as age, sex, handedness, and intelligence quotient, may affect all variables in the diagram. The FC example of the outcome with intra-subject processing defined in Definition~\ref{def:Y_def} is shown in the following Example~\ref{example1}. 

\begin{example}\label{example1}
(Functional Connectivity with Intra-subject Processing) The subject-level outcome $\bY_{FC}(a,m) \in \tR^{J}$ is defined as the vectorized Fisher z-transformed Pearson correlations computed from the residuals of $\bX_t(a,m)$ after intra-subject nuisance regression on $\bH_t(a,m)$:
\begin{equation}\label{eq:FC_def}
\bY_{FC}(a,m) = \text{atanh}\big[ \text{Corr}\left\{ \bX_t(a,m) - \mathbf{f}_{r}(\bH_t(a,m)) \mid \bU_M(a), \bS_M(a), \bU_B(a), \bS_B(a), \bW \right\} \big],
\end{equation}
where $\bm{f}_{r}$ denotes the intra-subject nuisance regression functions that characterize transient in-scanner motion artifact. The corresponding outcome function is given by $g(\cdot) = \text{atanh}[\text{Corr}\{\cdot\}]$. 
\end{example}

The motion related latent factors $\bU_M(a)$ and $\bS_M(a)$ are included in the class of indirect latent factors $\bU_I(a)$, while the brain-related latent factors $\bU_B(a)$ and $\bS_B(a)$ belong to the class of direct latent factors $\bU_D(a)$. For notational simplicity, we use $\bU_D(a)$ and $\bU_I(a)$ throughout this paper.

The average treatment effect (ATE) is defined as $\btau^{\textrm{ATE}} = \tE[\bY(1, M(1)) - \bY(0, M(0))]$, and represents the overall causal effect of the treatment, with the mediator taking its natural value under each treatment level. Our primary estimand is the natural direct effect (NDE), which characterizes the component of the treatment effect on the outcome $\bY$ that operates independently of the mediator $M$. For $a, a' \in \{0,1\}$, we define $\bY(a,a') = \bY(a, M(a'))$. The NDE is then given by
\begin{equation}\label{eq:def_NDE}
\bm{\tau}^{\text{NDE}} = \tE[\bY(1, M(0)) - \bY(0, M(0))],
\end{equation}
which contrasts potential outcomes under $A = 1$ and $A = 0$ while fixing the mediator at its natural value under $A = 0$. The natural indirect effect (NIE) is defined as $\bm{\tau}^{\text{NIE}} = \tE[\bY(1, M(1)) - \bY(1, M(0))]$, and represents the component of the treatment effect transmitted through the mediator $M$. For notational convenience, let $\bm{\psi}(a,a') = \tE[\bY(a, M(a'))]$. With this notation, the NDE, NIE, and ATE can all be expressed as contrasts of $\bm{\psi}(a,a')$ evaluated at appropriate values of $a$ and $a'$.

\subsection{Assumptions and Identification}
\label{sec:intra_identify}
Since the outcome of interest $\bY$ is not directly observable, existing derived outcome frameworks impose consistency assumptions on the full set of intra-subject observations $\bX$ \citep{qiu2023unveiling,du2025causal}. We extend this requirement to the complete intra-subject measurements, including both $\bX(a,m)$ and $\bH(a,m)$, in settings with intra-subject processing and a mediation structure.

\begin{assumption}\label{level2_consistency}
(Consistency). 
(i) $\bX(a,m) = \bX$ and $\bH(a,m) = \bH$ when $A = a$ and $M = m$, for $a \in \{0,1\}$ and all $m$;
(ii) $M(a) = M$ when $A = a$.
\end{assumption}

Since $\bY(a,m)$ is a deterministic function of $\bX(a,m)$ and $\bH(a,m)$, it follows that $\bY = \bY(a,m)$ whenever $A = a$ and $M = m$. \cref{level2_consistency} corresponds to the stable unit treatment value assumption (SUTVA), which ensures the absence of interference across subjects and guarantees that the treatment is well defined for all individuals.

\begin{assumption}\label{level2_positivity}
(Positivity).  
For all $m$ and $a \in \{0,1\}$, if $\tP(\bW = \bw) > 0$, then $0 < \tP(A = a \mid \bW = \bw) < 1$ and $0 < \tP(M = m \mid A = a, \bW = \bw) < 1$.
\end{assumption}

\begin{assumption}\label{level2_cond_ran}
(Conditional randomization). For $a \in \{ 0, 1 \}$, and all $m$,
(i) $\{\bX(a,m), \bH(a,m)\} \perp A \mid \bW$;
(ii) $\{ \bX(a,m), \bH(a,m) \} \perp M \mid A, \bW$;
(iii) $M(a) \perp A \mid \bW$.
\end{assumption}

\cref{level2_positivity} guarantees that each subject has a positive probability of receiving every level of the treatment and the mediator, which ensures that the proposed estimators are well defined. \cref{level2_cond_ran} asserts the absence of unmeasured confounding between the intra-subject measurements $\{\bX, \bH\}$ and the treatment $A$, between $\{\bX, \bH\}$ and the mediator $M$, and between the mediator $M$ and the treatment $A$. To identify the $\bm{\psi}(a,a')$, we further impose the following cross-world assumption.

\begin{assumption}\label{level2_cross_world}
(Cross-world independence). $\{ \bX(a,m), \bH(a,m) \} \perp M(a') \mid \bW$, for all $m$ and $a, a' \in \{0,1\}$.
\end{assumption}

The cross-world \cref{level2_cross_world} ensures that, conditional on the confounders $\bW$, the value of the mediator under $A = a'$ provides no information about the potential intra-subject measurements under $A = a$. Unlike classical causal inference settings, the proposed framework for intra-subject processed derived outcomes imposes conditional independence and cross-world assumptions on the full set of intra-subject response and nuisance measurements, rather than directly on the subject-level outcome \citep{tchetgen2012semiparametric}. We next introduce the derived outcome $\hbY(a,m)$, which is constructed as a function of observed variables with intra-subject processing.
\begin{definition}\label{def:Yhat_def}
(Derived outcome with intra-subject processing).
The derived outcome $\hbY(a,m)$ of $\bY(a,m)$ with intra-subject processing takes the form
\begin{equation}
\hbY(a,m) = \hat{g} \left\{ \bX(a,m) - \hat{\mathbf{f}}(\bH(a,m)) \right\},
\end{equation}
where $\hat{\mathbf{f}} = \{ \hat{f}_v: \tR^{p \times T} \to \tR^T, v = 1, \ldots, V \}$ denotes the fitted intra-subject processing functions, and $\hat{g}: \tR^{V \times T} \to \tR^J$ is an estimator of the pre-specified outcome function $g$.
\end{definition}

Note that the derived outcome $\hbY(a,m)$ in \cref{def:Yhat_def} contains two components of approximation. The function $\hat{g}$ approximates the known outcome function $g$, which typically involves a population expectation. Its form is specified in advance and is usually taken as the sample analogue of $g$, rather than being estimated in a data-adaptive manner. In contrast, the fitted $\hat{\mathbf{f}}$ estimates the unknown intra-subject processing function $\mathbf{f}$, which belongs to a broad semiparametric class and must be learned from intra-subject measurements. Taken together, the intra-subject processed derived outcome $\hbY(a,m)$ serves as a proxy for the potential outcome defined in \cref{def:Y_def}. For Example~\ref{example1} in our motivating rs-fMRI study, the derived FC outcome takes the form
\begin{equation}\label{eq:derived_FC_def}
\hbY_{FC}(a,m) = \text{atanh}\big[ \widehat{\text{Corr}}\left\{\bX(a,m) - \hat{\mathbf{f}}_{r}(\bH(a,m)) \right\} \big],
\end{equation}
where $\hat{\mathbf{f}}_{r}$ denotes the fitted nuisance regression models, and $\hat{g}(\cdot) = \text{atanh}[\widehat{\text{Corr}}\{\cdot\}]$ is a fixed estimator. 

We quantify the discrepancy between the derived and true outcomes as
\begin{equation}\label{eq:Delta_define}
\bm{\Delta}_T(a,m) = \tE[\hbY(a,m) \mid \bU_D(a), \bU_I(a), \bW] - \bY(a,m),
\end{equation}
where $\bm{\Delta}_T(a,m) = \{ \Delta_{jT}(a,m); 1 \le j \le J \}$ denotes the element-wise bias of the derived outcomes. The magnitude of this bias depends on the length of the intra-subject measurements $T$. For notational convenience, define $\hbY(a,a') = \hbY(a, M(a'))$ and $\bm{\Delta}(a,a') = \bm{\Delta}(a, M(a'))$ for $a,a' \in \{0,1\}$. We next state the asymptotic unbiasedness condition for identification.

\begin{assumption}\label{condition1_asy_unbias}
(Asymptotic unbiasedness).
The derived outcome $\hbY(a,a')$ is asymptotically unbiased for $\bY(a,a')$ if it satisfies $\underset{1 \le j \le J}{\max} |\tE [ \Delta_{jT}(a,a') ]| = o(1)$ as $T \to \infty$, for $a,a' \in \{0,1\}$. 
\end{assumption}

When $\hbY(a,a')$ satisfies the asymptotic unbiasedness condition in \cref{condition1_asy_unbias}, the identification of $\bm{\psi}(a,a')$ can be established as stated in the following \cref{thm:identification_NDE}.

\begin{thm}\label{thm:identification_NDE}
(Identification).
Suppose there exists an estimator $\hbY(a,a')$ that satisfies \cref{condition1_asy_unbias}. Under Assumptions \labelcref{level2_consistency,level2_cond_ran,level2_positivity,level2_cross_world}, the counterfactual parameter $\bm{\psi}(a,a')$ is identified as
\begin{equation}
    \tE\left[\tE\left\{ \tE \left( \hbY \mid A = a, M, \bW \right) \mid A = a', \bW\right\}\right].
\end{equation}
\end{thm}
The proof of \cref{thm:identification_NDE} is provided in the Web Supplement Section A.1. 

\section{Multiply Robust Estimation}
\label{sec:multi_robust}

Semiparametric theory for $\bm{\psi}(a,a')$ in settings where the outcome $\bY$ is directly observed has been developed in the causal mediation literature \citep{tchetgen2012semiparametric,vanderweele2015explanation}. However, the corresponding one step estimators and influence functions do not directly extend to the present setting, since $\bY$ is unobservable and is constructed from residuals obtained through intra-subject processing. \cref{thm:identification_NDE} shows that $\bm{\psi}(a,a')$ can be identified through the derived outcome $\hbY$. This result motivates the use of $\hbY$ as the basis for robust estimation in our framework.

We define the derived outcome–based target parameter as $\bm{\hat{\psi}}(a,a') = \tE[\hbY(a, M(a'))]$. For its $j$th component, we can construct a one-step estimator $\hat{\psi}_{n,j}(a,a')$ and decompose the estimation error as
\begin{equation}\label{eq:devY_error_decomp}
\hat{\psi}_{n,j}(a,a') - \psi_j(a,a') = \tP_n [\hat{\phi}_j(\tP)] + (\tP_n - \tP)\{\hat{\phi}_j(\hat{\tP}) - \hat{\phi}_j(\tP)\} + \hat{R}_2(\hat{\tP},\tP) + \tE[ \Delta_{jT}(a,a') ],
\end{equation}
where $\hat{\tP}$ is an estimate of $\tP$, and $\hat{\phi}_j(\tP)$ denotes the influence function constructed using the derived outcome $\hbY$, satisfying $\tP[\hat{\phi}_j(\tP)] = 0$ and $\tP[\hat{\phi}_j^2(\tP)] < \infty$. Details of derivations are provided in the Web Supplement Section A.2. The first term, after scaling by $\sqrt{n}$, converges in distribution to a normal limit by the central limit theorem. The second empirical process term is asymptotically negligible under a Donsker condition or when sample splitting is employed \citep{van2000asymptotic,chernozhukov2018double,kennedy2020sharp}. The second order remainder term $\hat{R}_2(\hat{\tP},\tP)$ is typically of order $o_{\tP}(n^{-1/2})$ under sufficient regularity and nuisance convergence rate conditions. The final bias term arises from using the derived outcome in place of the unobservable outcome and reflects the additional error introduced by intra-subject processing, which necessitates further regulating conditions for valid inference.

Two sources of uncertainty contribute to the bias term $\Delta_{jT}$ when approximate $Y_j(a,m)$ with $\hat{Y}_j(a,m)$, namely the error from using the outcome function estimator and the error from estimating the intra-subject processing function. To make this decomposition explicit, we introduce an intermediate quantity $\Tilde{\bY}(a,m) = \hat{g}\{ \bX(a,m) - \textbf{f}(\bH(a,m)) \}$, which applies the outcome function estimator $\hat{g}$ to residuals obtained using the true intra-subject processing function $\textbf{f}$. The bias term can then be decomposed into two corresponding components,
\begin{equation}\label{eq:two_parts}
\begin{aligned}
\Delta_{jT}(a,m) 
&= \tE[\hat{Y}_{j}(a,m) \mid \bU_D(a),\bU_I(a),\bW] - \tE[\Tilde{Y}_{j}(a,m) \mid \bU_D(a),\bU_I(a), \bW]\\
&+ \tE[\Tilde{Y}_{j}(a,m) \mid \bU_D(a),\bU_I(a),\bW] -  Y_{j}(a,m).
\end{aligned}
\end{equation}
The first difference, $\tE[\hat{Y}_{j}(a,m) \mid \bU_D(a), \bU_I(a), \bW] - \tE[\Tilde{Y}_{j}(a,m) \mid \bU_D(a), \bU_I(a), \bW]$, reflects the error introduced by replacing the true intra-subject processing function $\textbf{f}$ with its fitted counterpart. The second difference, $\tE[\Tilde{Y}_{j}(a,m) \mid \bU_D(a), \bU_I(a), \bW] - Y_{j}(a,m)$, captures the approximation error resulting from using the estimator $\hat{g}$ in place of the true outcome function $g$. The following \cref{condition2_asy_rootn} imposes a requirement involving the second term.

\begin{assumption}\label{condition2_asy_rootn}
Let $\Tilde{\Delta}_{jT}(a,m) = \tE[\Tilde{Y}_{j}(a,m) \mid \bU_D(a),\bU_I(a),\bW] -  Y_{j}(a,m)$, we have \\
$\underset{1 \le j \le J}{\max} |\tE[ \Tilde{\Delta}_{jT}(a,a') ]| = o(n^{-1/2})$, for $a \in \{0,1\}$, as $T \to \infty$.
\end{assumption}
\cref{condition2_asy_rootn} requires that the bias between $\Tilde{Y}_j(a,a')$ and $Y_j(a,a')$ vanishes at the $n^{-1/2}$ rate. This condition matches Condition 4 in \cite{qiu2023unveiling} and conditions stated in Corollary 6 and Proposition 10 of \cite{du2025causal}. It is mild and can be satisfied by a broad class of estimators $\hat{g}$ for $g$ in biomedical applications, including settings where $\bX(a,m) - \textbf{f}(\bH(a,m))$ are conditionally independent given $\{\bU_D(a), \bU_I(a), \bW\}$ and settings where $\bX(a,m) - \textbf{f}(\bH(a,m))$ exhibits weak temporal dependence \citep{qiu2023unveiling,du2025causal}.

The outcome function $g$ in \cref{def:Y_def} maps the $V$-dimensional intra-subject processed residuals to a $J$-dimensional subject-level outcome. In many biomedical applications, subject-level outcomes are defined through relationships among elements of the processed residuals $\bX_t(a,m) - \mathbf{f}(\bH_t(a,m))$. Examples include functional connections, co-activation patterns, and synchronization measures, which often represent the primary scientific targets. As a result, the outcome function $g$ frequently depends on cross-products of the residual components.

\begin{definition}\label{def:cross_product}
(Cross-product type outcome with intra-subject processing).
An outcome with intra-subject processing is said to be of cross-product type if each component of the subject-level outcome $\bY(a,m)$ corresponds to a cross-product of two residual components. Specifically, the $j$th element of $\bY(a,m)$, denoted by $Y_j(a,m) = Y_{vv'}(a,m)$ for some pair $(v,v')$ with $v,v' \in \{1,\ldots,V\}$, takes the form
\begin{equation}
g \left[ \tE\left\{ \bX_{tv}(a,m) - f_v(\bH_t(a,m)) \right\}
\left\{ \bX_{tv'}(a,m) - f_{v'}(\bH_t(a,m)) \right\}
\mid \bU_D(a), \bU_I(a), \bW \right].
\end{equation}
\end{definition}

The cross-product type derived outcome $\hbY(a,m)$ and the corresponding intermediate functional $\Tilde{\bY}(a,m)$ are defined analogously. In the presence of intra-subject processing, additional control is needed for the first term in the decomposition in \cref{eq:two_parts}. This requirement motivates the following assumption.

\begin{assumption}\label{condition3_intra_rate}
For cross-product type $\hbY$ and $\Tilde{\bY}$, let $\hat{\mathbf{f}} = \{ \hat{f}_v;\ v = 1, \ldots, V \}$ denote the fitted intra-subject processing functions. We assume that the intra-subject processing satisfies $\underset{1 \le v \le V}{\max} \| \hat{f}_v - f_v \|_2 = \cO(T^{-\alpha_0})$, for some $\alpha_0 \in (0, 1/2]$ such that $n^{1/2} T^{-2\alpha_0} = o(1)$, as $n, T \to \infty$.
\end{assumption}

This condition specifies the convergence rate required for the fitted intra-subject processing functions so that the bias from estimating $\mathbf{f}$ is asymptotically negligible at the inter-subject level. When the true intra-subject processing function $f_v$ is unknown, its estimator $\hat{f}_v$ must be consistent. The framework also includes settings in which the intra-subject processing functions are known, in which case \cref{condition3_intra_rate} is no longer required since $\underset{1 \le v \le V}{\max} \| \hat{f}_v - f_v \|_2 = 0$.

For cross-product type outcomes, the induced bias term $\tE[\hat{Y}_{j}(a,m) \mid \bU_D(a), \bU_I(a), \bW] - \tE[\Tilde{Y}_{j}(a,m) \mid \bU_D(a), \bU_I(a), \bW]$ is of order $T^{-2\alpha_0}$, as shown in the Web Supplement Section A.3. This rate makes it more likely that the condition $n^{1/2} T^{-2\alpha_0} = o(1)$ holds, even when the intra-subject processing functions are estimated at rates slower than $T^{-1/2}$. This feature allows for flexible estimation in settings with temporal dependence or complex non-linear relationships between intra-subject responses and nuisance variables. As a result, data-driven machine learning methods can be used to reduce model misspecification and obtain consistent estimates of the intra-subject processing functions.

For outcomes that are not of cross-product type, the required rate in \cref{condition3_intra_rate} can differ and depends on the structure of the outcome function $g$. For example, the rs-fMRI metric Amplitude of Low-Frequency Fluctuations measures the strength of spontaneous low-frequency BOLD signal oscillations at each brain location and does not involve interactions across locations. In this case, the intra-subject processing rate condition becomes $n^{1/2} T^{-\alpha_0} = o(1)$ as $n, T \to \infty$.

\begin{lemma}\label{lemma:delta_rootn}
Suppose the fitted intra-subject processing functions $\hat{\mathbf{f}}$ and the intermediate estimate $\Tilde{\bY}(a,a')$ satisfy \cref{condition2_asy_rootn} and \cref{condition3_intra_rate}. Then the bias between the derived outcome $\hbY(a,a')$ and the true outcome $\bY(a,a')$ satisfies $\underset{1 \le j \le J}{\max} |\tE\left[ \Delta_{jT}(a,a') \right]| = o(n^{-1/2})$ as $n,T \to \infty$, for $a,a' \in \{0,1\}$.
\end{lemma}

The proof of \cref{lemma:delta_rootn} is given in the Web Supplement Section A.3. \cref{lemma:delta_rootn} establishes an $n^{1/2}$ convergence rate for $\Delta_{jT}(a,a')$, which is stronger than the asymptotic unbiasedness condition in \cref{condition1_asy_unbias} required for identification. In our motivating Example~\ref{example1}, we use the Fisher z-transformed sample Pearson correlation as the estimator $\hat{g}$ for the outcome function $g$. Under $\beta$-mixing and sub-Gaussian moment conditions, Proposition S2 of \cite{qiu2023unveiling} shows that $\Tilde{\bY}_{FC}(a,a')$ satisfies \cref{condition2_asy_rootn} when $n^{1/2} \log J = o(T)$. In addition, \cref{condition3_intra_rate} can be satisfied by flexibly regressing out intra-subject nuisance effects using semiparametric machine learning methods.

\cref{lemma:delta_rootn} implies that the term $\tE[\Delta_{jT}]$ in \cref{eq:devY_error_decomp} is asymptotically negligible under the stated conditions, which justifies the use of influence functions constructed on the derived outcomes for statistical inference. We therefore consider the one-step augmented inverse probability weighted (AIPW) estimator for the $j$th component of $\bm{\hat{\psi}}(a,a')$, given by
\begin{equation}
\begin{aligned}
    \hat{\psi}_{n,j}^{\text{AIPW}}(a,a') &= \tP_n \Big[\frac{\mathbb{I}(A=a)}{\hat{\pi}_A(a | \bW)} \frac{\hat{\pi}_{M}(M | a',\bW )}{\hat{\pi}_{M}(M | a,\bW )} \left\{\hat{Y}_{j}-\hat{b}_{j}(M, a,\bW)\right\}\\
    &+ \frac{\mathbb{I}(A=a')}{\hat{\pi}_{A}(a' | \bW)}\left\{ \hat{b}_{j}(M,a,\bW) - \hat{\xi}_{a a' j}(\bW) \right\} + \hat{\xi}_{a a' j}(\bW) \Big],
\end{aligned}
\end{equation}
where $\pi_A(a | \bw) = \tP(A = a | \bW = \bw)$ denotes the propensity score, $\pi_M(m | a, \bw) = \tP(M = m | A = a, \bW = \bw)$ denotes the mediator density, and $b_j(m,a,\bw) = \tE[\hat{Y}_j | M = m, A = a, \bW = \bw]$ denotes the outcome regression. We further define $\xi_{aa'j}(\bw) = \int b_j(m,a,\bw) \pi_M(m | a', \bw) dm$. The collection $\{ \hat{\pi}_A, \hat{\pi}_M, \hat{b}_j, \hat{\xi}_{aa'j} \}$ denotes estimators of the corresponding nuisance functions $\{ \pi_A, \pi_M, b_j, \xi_{aa'j} \}$. Direct estimation of the density of a continuous mediator is often challenging. To avoid explicit estimation of $\pi_M(m \mid a,\bw)$, we combine Bayes’ rule with a sequential regression representation. By Bayes' rule, 
$$\frac{\pi_{M}(m | a',\bw)}{\pi_{M}(m | a, \bw)} = \frac{\pi_{A}(a' | m, \bw)\pi_A(a | \bw)}{\pi_A(a | m, \bw)\pi_{A}(a' | \bw)}.$$
In addition, the functional $\xi_{aa'j}(\bW)$ admits the sequential regression form $\xi_{aa'j}^{sr}(\bW) = \tE[b_{j}(m,a,\bW) \mid A = a', \bW = \bw]$, which can be estimated without evaluating $\pi_M(m | a,\bw)$. Using these representations, the AIPW estimator can be rewritten as
\begin{equation}\label{eq:NDE_AIPW}
\begin{aligned}
    \hat{\psi}_{n,j}^{\text{AIPW}}(a,a') &= \tP_n \Big [ \frac{\mathbb{I}(A=a)}{\hat{\pi}_A(a | \bW)} \frac{\hat{\pi}_{A}(a' | M, \bW)\hat{\pi}_A(a | \bW)}{\hat{\pi}_A(a | M, \bW)\hat{\pi}_{A}(a' | \bW)} \left\{\hat{Y}_{j}-\hat{b}_{j}(M,a,\bW) \right\}\\
    &+ \frac{\mathbb{I}(A=a')}{\hat{\pi}_{A}(a' | \bW )}\left\{ \hat{b}_{j}(M,a,\bW ) - \hat{\xi}_{aa'j}^{sr}(\bW) \right\} + \hat{\xi}_{aa'j}^{sr}(\bW) \Big],
\end{aligned}
\end{equation}
and the asymptotic expansion of $\hat{\psi}_{n,j}^{\text{AIPW}}(a,a')$ is provided in the following \cref{thm:asy_expansion}.

\begin{thm}\label{thm:asy_expansion}
Let $O_1,\ldots,O_n$ be i.i.d. observations with $O_i = (A_i,\bW_i,M_i,\bX_i,\bH_i)$. Suppose there exists an intra-subject processed derived outcome $\hbY$ satisfying Assumptions \labelcref{condition2_asy_rootn,condition3_intra_rate}. Under Assumptions \labelcref{level2_consistency,level2_cond_ran,level2_positivity,level2_cross_world}, assume the following conditions hold for $a,a' \in \{0,1\}$ and $j=1,\ldots,J$.\\
(i). (Boundedness). There exists $c_0 \in (0,1)$ and positive constants $c_1, c_2, c_3, c_4$, such that $\pi_A(a|\bW), \hat{\pi}_A(a|\bW) \in [c_0, 1-c_0]$, $\underset{1 \leq j \leq J}{\max}\{ |Y_j|, |\hat{Y}_j| \} < c_1$, $\underset{1 \leq j \leq J}{\max}\{ \| b_j\|_{\infty}, \|\hat{b}_j \|_{\infty} \} < c_2$, $\underset{1 \leq j \leq J}{\max}\{ \| \xi_{aa'j}^{sr}\|_{\infty}, \| \hat{\xi}_{aa'j}^{sr} \|_{\infty} \} < c_3$, $\max \{ \| \frac{\pi_{A}(a' | M, \bW)\pi_A(a | \bW)}{\pi_A(a | M, \bW)\pi_{A}(a' | \bW)} \|_{\infty}, \| \frac{\hat{\pi}_{A}(a' | M, \bW)\hat{\pi}_A(a | \bW)}{\hat{\pi}_A(a | M, \bW)\hat{\pi}_{A}(a' | \bW)} \|_{\infty} \} < c_4$.\\
(ii). (Nuisance rates). The nuisance estimators satisfy $\| \pi_{A}(a |m, \bw)\pi_A(a|\bw) - \hat{\pi}_{A}(a | m, \bw)\hat{\pi}_A(a | \bw) \|_2 = \cO(n^{-\alpha})$, $\| b_{j}(m,a,\bw) - \hat{b}_{j}(m,a,\bw) \|_2 = \cO(n^{-\beta})$, $\| \pi_{A}(a|\bw) - \hat{\pi}_A(a | \bw) \|_2 = \cO(n^{-\gamma})$, and $\| \xi_{aa'j}^{sr}(\bw) - \hat{\xi}_{aa'j}^{sr}(\bw) \|_2 = \cO(n^{-\zeta})$, where $\alpha + \beta > \frac{1}{2}$, $\beta + \gamma > \frac{1}{2}$, and $\gamma + \zeta > \frac{1}{2}$.\\
Then, as $n,T,J \rightarrow \infty$, the AIPW estimator defined in \cref{eq:NDE_AIPW} admits the expansion
$( \hat{\psi}_{n,j}^{\text{AIPW}}(a,a') - \psi_j(a,a') ) = \tP_n [\phi_{aa'j}(O)] + o_{\tP}(n^{-1/2})$. Where the influence function $\phi_{aa'j}(O_i)$ is given by
\begin{equation}
\begin{aligned}
    \phi_{aa'j}(O_i) &= \frac{\mathbb{I}(A_i=a)}{\pi_A(a | \bW_i)} \frac{\pi_{A}(a' | M_i, \bW_i)\pi_A(a | \bW_i)}{\pi_A(a | M_i, \bW_i)\pi_{A}(a' | \bW_i)} \left(\hat{Y}_{ij}-b_j(M_i,a,\bW_i)\right)\\
    &+ \frac{\mathbb{I}(A_i=a')}{\pi_{A}(a' | \bW_i )}\left( b_j\left(M_i,a,\bW_i\right) - \xi_j^{sr}(a,a',\bW_i) \right) + \xi_j^{sr}(a,a',\bW_i) - \psi_j(a,a').
\end{aligned}
\end{equation}
\end{thm}

The proof of \cref{thm:asy_expansion} is provided in the Web Supplement Section A.3. The propensity score models $\hat{\pi}_A(a \mid \bw)$ and $\hat{\pi}_A(a \mid m,\bw)$ are estimated by fitting the treatment indicator $A$ on $\bW$ or on $(M,\bW)$, respectively, which follows the same strategy as in causal mediation analysis with observable outcomes. The outcome regression $\hat{b}_j(m,a,\bw)$ is obtained by regressing the intra-subject processed derived outcome $\hbY$ on $(A,M,\bW)$, since the target function $b_j(m,a,\bw)$ is defined in terms of the derived outcome rather than the unobservable outcome $\bY$. This differs from classical causal inference settings, where outcome regression is specified directly for $\bY$. The quantity $\hat{\xi}_j^{sr}(a,a',\bW)$ is estimated using a sequential regression approach.The $\hat{b}_j(M,a,\bW)$ is first regressed on $(A,\bW)$, and the resulting fitted model is then evaluated at $(a',\bW)$.

The nuisance rate conditions in \cref{thm:asy_expansion} guarantee that the second order remainder term satisfies $\hat{R}_2(\hat{\tP},\tP) = o_{\tP}(n^{-1/2})$. These cross-product type constraints allow the inter-subject nuisance estimators to converge at rates slower than $n^{-1/2}$, enabling the use of a broad class of flexible statistical and machine learning methods. To ensure the empirical process term in the semiparametric von Mises expansion in \cref{eq:devY_error_decomp} is asymptotically negligible, we employ cross fitting when estimating the nuisance components $\{ \hat{\pi}_A(a \mid \bw), \hat{\pi}_A(a \mid m, \bw), \hat{b}_j(m,a,\bw), \hat{\xi}_j^{sr}(a,a',\bw) \}$. This strategy removes the need for Donsker conditions and improves robustness. As an immediate consequence of \cref{thm:asy_expansion}, the estimator $\hat{\psi}_{nj}^{\text{AIPW}}(a,a')$ exhibits multiply robust properties.

\begin{corollary}\label{corollary:multiply_robust}
(Multiple robustness) 
The AIPW estimator $\hat{\psi}_{n,j}^{\text{AIPW}}(a,a')$ defined in \cref{thm:asy_expansion} is consistent for $\psi_j(a,a')$ if any one of the following three conditions is satisfied:\\
(i). $\hat{\pi}_{A}(a | m, \bw) \overset{P}{\rightarrow} \pi_{A}(a | m, \bw)$ and $\hat{\pi}_A(a | \bw) \overset{P}{\rightarrow} \pi_A(a | \bw)$;\\
(ii). $\hat{b}_j(m,a,\bw) \overset{P}{\rightarrow} b_j(m,a,\bw)$ and $\hat{\pi}_A(a | \bw) \overset{P}{\rightarrow} \pi_A(a | \bw)$;\\
(iii). $\hat{b}_j(m,a,\bw) \overset{P}{\rightarrow} b_j(m,a,\bw)$ and $\hat{\xi}_j^{sr}(a,a',\bW_i) \overset{P}{\rightarrow} \xi_j^{sr}(a,a',\bW_i)$.
\end{corollary}

The proof of \cref{corollary:multiply_robust} is given in the Web Supplement Section A.3. Compared with classical causal inference settings, the presence of intra-subject processing requires additional conditions, namely Assumptions \labelcref{condition2_asy_rootn,condition3_intra_rate}, to obtain an asymptotic linear representation for $\hat{\psi}_{nj}^{\text{AIPW}}(a,a')$. When these intra-subject level conditions are satisfied, multiple robustness properties can be established at the inter-subject level. To reduce the risk of model misspecification, we use flexible machine learning methods to estimate both intra-subject processing functions and inter-subject nuisance models. In particular, we adopt the Super Learner framework, an ensemble learning approach that combines multiple candidate algorithms to improve predictive performance \citep{breiman1996stacked,van2007super}. Super Learner forms an optimally weighted combination of pre-specified parametric and nonparametric learners based on cross-validated prediction error. The resulting estimator is asymptotically guaranteed to perform at least as well as the best learner in the library, supporting accurate estimation of both intra- and inter-subject components and enhancing overall robustness.

\section{Simultaneous Inference}
\label{sec:simultaneous}

By the central limit theorem, \cref{thm:asy_expansion} directly yields asymptotic normality of the individual AIPW estimators, formalized in \cref{corollary:asy_norm}.

\begin{corollary}\label{corollary:asy_norm}
(Asymptotic normality)
Under the conditions of \cref{thm:asy_expansion}, assume that $\sigma^2_{aa'j} = \tP [\phi^2_{aa'j}(O)] \ge c_5$ for some constant $c_5 > 0$ and for $j = 1,\ldots,J$. Then we have $\sqrt{n}\left( \hat{\psi}_{n,j}^{\text{AIPW}}(a,a') - \psi_j(a,a') \right) \overset{d}{\rightarrow} N(0, \sigma^2_{aa'j}).$
\end{corollary}

The proof of \cref{corollary:asy_norm} is given in the Web Supplement Section A.4. The $j$th component of the AIPW estimators, $\hat{\tau}_{n,j}^{\mathrm{AIPW}} \in \{ \hat{\tau}_{n,j,\mathrm{AIPW}}^{\mathrm{NDE}}, \hat{\tau}_{n,j,\mathrm{AIPW}}^{\mathrm{NIE}}, \hat{\tau}_{n,j,\mathrm{AIPW}}^{\mathrm{ATE}} \}$, can be obtained directly as a contrast of $\hat{\psi}_{n,j}^{\text{AIPW}}(a,a')$. The corresponding influence function is denoted by $\eta_j \in \{ \eta_j^{\mathrm{NDE}} = \phi_{10j} - \phi_{00j}, \eta_j^{\mathrm{NIE}} = \phi_{11j} - \phi_{10j}, \eta_j^{\mathrm{ATE}} = \phi_{11j} - \phi_{00j} \}$, with variance $\theta_j^2 = \tP [\eta_j^2(O)]$. As a direct consequence of \cref{corollary:asy_norm}, it follows that $\sqrt{n} (\hat{\tau}_j - \tau_j) \overset{d}{\rightarrow} \cN(0,\theta_j^2)$.

Following \cite{du2025causal}, consistent variance estimators $\hat{\sigma}^2_{aa'j} = \tP_n [\hat{\phi}^2_{aa'j}(O)]$ for $\sigma^2_{aa'j}$ can be constructed using the estimated influence functions $\hat{\phi}_{aa'j}(O_i)$ with cross-fitting. Therefore, the variance estimator of $\hat{\tau}_{n,j}^{\mathrm{AIPW}}$ is $\hat{\theta}_j^2 = \tP_n [\hat{\eta}^2_{aa'j}(O)]$. This facilitates hypothesis testing for $H_{0j}:\tau_j = 0$ using the test statistic $t_j = n^{-1/2}\hat{\tau}_j/\hat{\theta}_j$, as well as construction of $(1-\alpha)$ confidence intervals of the form $\hat{\tau}_j \pm n^{-1/2} z_{1-\alpha/2}\hat{\theta}_j$, where $z_{1-\alpha/2}$ denotes the $(1-\alpha/2)$ quantile of the standard normal distribution.

\subsection{Multiple Testing Control}
\label{sec:multi_test}

For high-dimensional inference, we adopt the Gaussian approximation and multiplier bootstrap results developed for the derived outcome framework in \cite{qiu2023unveiling} and \cite{du2025causal}. These results enable the construction of simultaneous inference intervals under mild regularity conditions, including exclusion of super-efficient components whose asymptotic variances vanish for which $\theta_j^2 \rightarrow 0$ \citep{chernozhukov2013gaussian,belloni2018high}. In practice, this is implemented by screening out estimators with small estimated variances $\hat{\theta}_j^2$ to construct an empirical informative set $\cS_1 = \{j: \hat{\theta}_j^2 \ge c_0 \}$, where $c_0 > 0$ is a small constant. Under bounded variance and weak dependence conditions, the distribution of the maximum standardized estimator over this set, $\underset{j \in \cS_1}{\max} \sqrt{n} |\hat{\tau}_j - \tau_j| \hat{\theta}_j^{-1}$, admits a Gaussian approximation, which enables the construction of simultaneous confidence intervals through a multiplier bootstrap procedure \citep{qiu2023unveiling,du2025causal}.

When the number of hypotheses is large, simultaneous confidence intervals can be conservative. We employ a multiple testing procedure that controls the false discovery proportion (FDP), defined as the ratio of false positives to the total number of discoveries. For a pre-specified tolerance level $c > 0$, the goal is to control the probability $\tP(\text{FDP} > c)$, which is referred to as the FDP exceedance rate (FDPex), at level $\alpha$. Following \citet{du2025causal}, this influence function based adaptive procedure combines a step-down algorithm with a Gaussian multiplier bootstrap and an augmentation step, yielding asymptotic control of $\tP(\text{FDP} > c) \le \alpha$. This approach accounts for the variability of the FDP within a single realized sample, leading to a stronger error control than methods targeting the false discovery rate \citep{genovese2006exceedance,qiu2023unveiling,du2025causal}.

In this work, these inferential procedures are applied to our proposed multiply robust estimators $\hat{\psi}_{n,j}^{\text{AIPW}}(a,a')$ for intra-subject processed outcomes. Theoretical validity follows from establishing the influence function representation in \cref{thm:asy_expansion} and associated regularity conditions required for Gaussian approximation remain satisfied in the presence of intra-subject processing. Additional implementation details and algorithmic steps are provided in the Web Supplement Section B.

\section{Simulations}
\label{sec:simulations}

We consider a simulation setting with one brain region serving as the seed region (region 0) and two additional regions (regions 1 and 2). The FC between the region 0 and region 1 has a nonzero NDE, while the FC between region 0 and region 2 has zero NDE. Independent and identically distributed confounders $\bW \in \tR^3$ are generated from a three dimensional multivariate normal distribution with mean zero and covariance matrix $\bm \Sigma_W = (\sigma_{W,j_1,j_2})_{3 \times 3}$, where $\sigma_{W,j_1,j_2} = 0.5^{|j_1 - j_2|}$. The treatment $A$ is generated from a logistic regression model with probability of treatment given by $P(A=1 \mid \bW) = 1/(1 + \exp(-0.2 + 0.4 \bone^\top_3 \bW))$. The latent motion factor is defined as $U_M = 0.5 + 0.5A + 0.1 \bone^\top \bW + \epsilon_U$, where $\epsilon_U \overset{iid}{\sim} N(0,0.4)$. The mediator $M$, representing the time invariant motion trait, is generated as $M = U_M + \epsilon_m$, with $\epsilon_m \overset{iid}{\sim} N(0,0.1)$. The intra-subject nuisance measurements $\bH \in \tR^{3 \times T}$ are generated as $\bH \overset{iid}{\sim} MVN\left\{(U/3)\bone_3,\bSigma_H\right\}$ with $\bSigma_H = (\sigma_{H,j_1,j_2}){3 \times 3}$ and $\sigma_{H,j_1,j_2} = 0.3^{|j_1 - j_2|}$. 

We define $Y_1 = 0.1 + 0.6A + M - 0.6A \times M + 0.1 \bone^\top \bW + \epsilon_S$ for the pair of region 1 and region 0 and define $Y_2 = 0.1 + 0.2A - 0.4A \times M + 0.1 \bone^\top \bW + \epsilon_S$ for the pair of region 2 and region 0, where $\epsilon_S \overset{iid}{\sim} N(0,0.5)$. Under this design, the true natural direct effect equals 0.3 and 0, respectively. To generate intra-subject rs-fMRI responses contaminated by nuisance measurements, we first simulate the time series innovations at region 0 as $E_{0,t}\overset{iid}{\sim}N(0,1)$ for $t = 1,\ldots,T$. The innovation series $E_{1,t}$ and $E_{2,t}$ for the other two regions are then generated from standard normal distributions, with their covariances with $E_{0,t}$ set to $\tanh(Y_1)$ and $\tanh(Y_2)$, respectively. We collect these error series into $\bE \in \tR^{3 \times T}$. For these simulations, we define framewise displacement as $\text{FD}_t = |H_{1t}| + |H_{2t}| + |H_{3t}|$ for $t = 1,\ldots,T$. For the first time point, the observed intra-subject response is defined as $\bX_{1} = 0.6 H_{11} \times H_{31} - 0.8 H_{21} \times H_{31} + 0.8 \mathbb{I}(\text{FD}_1 > 2) \times H_{11} \times H_{21} + \bE_{1}$. For $t = 2,\ldots,T$, the response is generated as $\bX_{t} = \rho \bE_{t-1} + 0.6 H_{11} \times H_{31} - 0.8 H_{21} \times H_{31} + 0.8 \mathbb{I}(\text{FD}_1 > 2) \times H_{11} \times H_{21} + \sqrt{1-\rho^2} \bE_{t-1}$, where $\rho = 0.3$ controls the strength of temporal dependence.

We considered sample sizes $n = 150$ and $300$, with numbers of time points $T = 300$ and $600$. Each simulation setting was replicated 1000 times. For intra-subject processing, we compared three nuisance regression approaches. The first approach, denoted as ``12p", uses linear regression with 12 nuisance regressors. These include 3 main motion terms, 3 temporal derivatives defined as first differences, and their 6 quadratic terms. This method mimics the commonly used ``36p" strategy in rs-fMRI studies, which applies linear motion regression with 9 nuisance variables and their 27 expanded terms \citep{ciric2017benchmarking}. The second approach, termed ``12p+Scrubbing", removes time points with framewise displacement ($\text{FD}_t$) greater than 3 and then applies the 12p regression to the remaining data. A subject is excluded if more than 35\% of time points are removed. This approach combines linear nuisance regression with volume censoring and reflects the widely used scrubbing strategy in the literature \citep{ciric2017benchmarking}. The third approach, labeled ``SL", applies an ensemble machine learning method using Super Learner to model intra-subject responses based on 3 main motion terms and 3 temporal derivatives \citep{van2007super}. The Super Learner library includes the baseline mean model, the generalized linear model, and multivariate adaptive regression splines models.

\begin{figure}[t]
    \centering
    \includegraphics[width=0.9\linewidth]{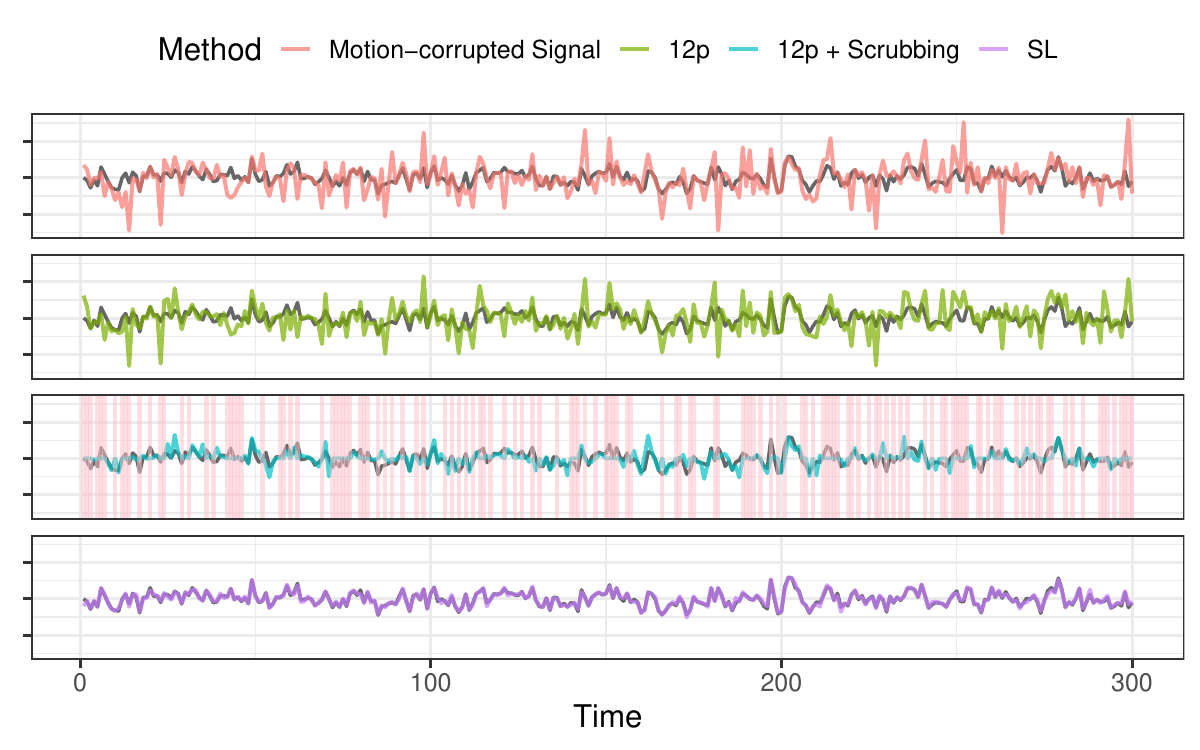}
    \caption{Example of intra-subject motion regression from a single simulation. The black line shows the simulated motion-free neural signal. In the first panel, the red line represents the motion-contaminated signal before correction. In the second panel, the green line (12p) shows residuals from linear regression with 3 main terms, 3 temporal derivatives, and 6 quadratic terms. In the third panel, the blue line (12p + Scrubbing) applies the same regression after removing high-motion time points, with censored points marked in pink. In the final panel, the purple line (SL) shows residuals obtained from Super Learner nuisance regression.}
    \label{fig:simu_results}
\end{figure}

We show an example of intra-subject processing from a single simulation replicate with $T = 300$ in \cref{fig:simu_results}. Three subject-level motion correction methods are compared. In each panel, the black line represents the simulated motion-free neural signal, which serves as the ground truth. In the first row, the red line shows the motion-contaminated signal before correction. The observed fMRI signal is strongly affected by head motion. In the second row, the green line labeled ``12p" shows residuals from linear regression using the 12p method. This approach reduces motion effects but leaves many motion spikes. In the third row, the blue line labeled ``12p + Scrubbing" further removes high-motion time points. The censored points are shown in pink. This step reduces motion spikes but removes many time points, leading to missing data at scrubbed locations. In the final row, the purple line labeled ``SL" shows residuals from Super Learner based machine learning nuisance regression. This method removes motion contamination more effectively and better recovers the underlying neural signal, with residuals closely matching the motion-free ground truth.

\begin{table}[t]
\centering
\caption{\label{tab:simulation_results}Simulation results based on 1000 repetitions comparing bias, variance, MSE, power, and type-I error across methods.}
\centering
\resizebox{\ifdim\width>\linewidth\linewidth\else\width\fi}{!}{
\fontsize{9}{11}\selectfont
\begin{tabular}[t]{cc l cccc cccccc l cccc cccccc l cccc cccccc l cccc cccccc l cccc cccccc l cccc cccccc l cccc cccccc l cccc cccccc l cccc cccccc l cccc cccccc l cccc cccc}
\toprule
\multicolumn{3}{c}{ } & \multicolumn{4}{c}{Region 1: Truth = 0.3} & \multicolumn{4}{c}{Region 2: Truth = 0} \\
\cmidrule(l{3pt}r{3pt}){4-7} \cmidrule(l{3pt}r{3pt}){8-11}
$n$ & $T$ & Method & Bias & SD & MSE & Power & Bias & SD & MSE & Type-I Error\\
\midrule
\multirow{5}{*}{150} &
\multirow{5}{*}{300} &
12p+Linear & 0.096 & \textbf{0.091} & \textbf{0.017} & \textbf{0.987} & -0.129 &\textbf{ 0.060} & 0.020 & 0.561\\
& & 12p+Linear M & -0.188 & 0.099 & 0.045 & 0.254 & -0.076 & 0.074 & \textbf{0.011} & 0.216\\
& & 12p Scrub+Linear & 0.115 & 0.100 & 0.023 & 0.984 & -0.122 & 0.077 & 0.021 & 0.386\\
& & 12p Scrub+Linear M & -0.162 & 0.106 & 0.037 & 0.300 & -0.082 & 0.088 & 0.015 & 0.172\\
& & SL+AIPW & \textbf{-0.040} & 0.151 & 0.024 & 0.489 & \textbf{-0.001} & 0.143 & 0.020 & \textbf{0.062}\\
\addlinespace
\multirow{5}{*}{150} &
\multirow{5}{*}{600} &
12p+Linear & 0.093 & \textbf{0.090} & \textbf{0.017} & \textbf{0.991} & -0.130 & \textbf{0.060} & 0.020 & 0.590\\
& & 12p+Linear M & -0.190 & 0.098 & 0.046 & 0.252 & -0.077 & 0.073 & \textbf{0.011} & 0.211\\
& & 12p Scrub+Linear & 0.107 & 0.101 & 0.021 & 0.980 & -0.125 & 0.076 & 0.021 & 0.406\\
& & 12p Scrub+Linear M & -0.168 & 0.108 & 0.040 & 0.286 & -0.086 & 0.089 & 0.015 & 0.190\\
& & SL+AIPW & \textbf{-0.028} & 0.165 & 0.028 & 0.491 & \textbf{-0.004} & 0.156 & 0.024 & \textbf{0.082}\\
\addlinespace
\multirow{5}{*}{300} &
\multirow{5}{*}{300} &
12p+Linear & 0.100 & \textbf{0.064} & 0.014 & \textbf{1.000} & -0.127 & \textbf{0.044} & 0.018 & 0.847\\
& & 12p+Linear M & -0.185 & 0.069 & 0.039 & 0.438 & -0.073 & 0.051 & \textbf{0.008} & 0.328\\
& & 12p Scrub+Linear & 0.118 & 0.071 & 0.019 & 1.000 & -0.120 & 0.054 & 0.017 & 0.601\\
& & 12p Scrub+Linear M & -0.159 & 0.075 & 0.031 & 0.515 & -0.080 & 0.062 & 0.010 & 0.268\\
& & SL+AIPW & \textbf{-0.049} & 0.093 & \textbf{0.011} & 0.813 & \textbf{-0.004} & 0.095 & 0.009 & \textbf{0.059}\\
\addlinespace
\multirow{5}{*}{300} &
\multirow{5}{*}{600} &
12p+Linear & 0.100 & \textbf{0.064} & 0.014 & \textbf{1.000} & -0.125 & 0.042 & 0.017 & 0.851\\
& & 12p+Linear M & -0.184 & 0.068 & 0.038 & 0.466 & -0.072 & 0.050 & \textbf{0.008} & 0.326\\
& & 12p Scrub+Linear & 0.115 & 0.071 & 0.018 & 1.000 & -0.119 & 0.052 & 0.017 & 0.638\\
& & 12p Scrub+Linear M & -0.162 & 0.073 & 0.031 & 0.535 & -0.081 & 0.060 & 0.010 & 0.291\\
& & SL+AIPW & \textbf{-0.028} & 0.103 & \textbf{0.011} & 0.823 & \textbf{0.001} & 0.101 & 0.010 & \textbf{0.066}\\
\bottomrule
\end{tabular}}
\end{table}

After intra-subject processing, FC is estimated by computing the Fisher z-transformed Pearson correlation between residual time series from two brain regions. At the inter-subject level, we applied three approaches. The first approach estimates treatment effects using linear regression with treatment and confounders $\bW$ (hereafter, ``Linear"). The second approach adds the motion variable $M$ and estimates the direct effect of A assuming a linear model (hereafter, ``Linear M"). The third approach is our proposed AIPW estimator, which treats motion as a mediator and uses Super Learner to estimate propensity scores and outcome models with cross-fitting (herafter, ``AIPW"). We evaluated five methods that combine intra-subject and inter-subject strategies: 12p + Linear, 12p + Linear M, 12p Scrubbing + Linear, 12p Scrubbing + Linear M, and SL + AIPW. We evaluated performance using bias, variance, and mean squared error (MSE). Power was assessed based on region 1, and type-I error was assessed based on region 2.

\cref{tab:simulation_results} summarizes results from 1000 simulations across different sample sizes and numbers of time points. Methods based on linear inter-subject models without accounting for motion mediation, including 12p + Linear and 12p Scrub + Linear, achieve low variance and high statistical power but suffer from substantial bias, which leads to severe inflation of the type-I error rate. Although the 12p + Linear method attains the smallest MSE at smaller sample sizes due to its low variance, its MSE exceeds that of the proposed SL + AIPW estimator when the sample size increases to $n = 300$. Including motion in the linear inter-subject models reduces bias and type-I error, but still fails to achieve valid error control and results in a drastic loss of power. In contrast, the SL + AIPW estimator shows small bias, low MSE, and stable performance across all settings. It achieves type-I error rates closest to the nominal level while maintaining reasonable power, exceeding 0.8 when $n = 300$, with improvements as both sample size and number of time points increase. These results highlight the need to jointly model intra-subject processing and inter-subject relationships using flexible methods to ensure valid causal inference.

\section{Application}
\label{sec:application}

We analyzed rs-fMRI data from autism children selected from the Autism Brain Imaging Data Exchange (ABIDE) dataset \citep{DiMartino2013TheAutism,di2017enhancing}. There were 128 ASD children (104 males) aged 8-13 in our study sample. Each participant has one T1w anatomical scan and an rs-fMRI scan with length varying from 5.5 to 6.5 minutes. The raw imaging data were preprocessed using a cortical surface fMRI pipeline with \texttt{fMRIPrep} \citep{esteban2019fmriprep}. We defined the treatment as any reported stimulant medications (see the Web Supplement Section C.2 for details). Among them, 101 ASD children (84 males) did not take stimulant medications, and 27 ASD children (20 males) took stimulant medications. Children were instructed to not take stimulants the day prior to and the day of the scan. Dosages and duration of stimulant treatment are heterogeneous. Our goal is to examine possible long-term impacts of such stimulant medications. Details on the dataset and preprocessing are in the Web Supplement Section C.

Within each subject, the mean signals from nuisance tissue compartments, including the global signal, white matter, and cerebrospinal fluid, as well as three translational and three rotational motion parameters were calculated by \texttt{fMRIprep} and used as our intra-subject nuisance measurements. The set of these nine variables, their first-order temporal derivatives, and squared terms constitutes the ``36p" model to be removed from raw intra-subject rs-fMRI data \citep{ciric2017benchmarking}, which were used in the 36p and linear model approach (see below). For the 36p with scrubbing and participant removal approach, we applied scrubbing in which time points with framewise displacement greater than 0.2 mm were removed, and the entire dataset for a subject was excluded if less than five minutes of fMRI data remained after scrubbing \citep{power2014methods}. In our proposed framework, we performed adaptive intra-subject processing without excluding participants using the \texttt{SuperLearner} package in \texttt{R} on the rs-fMRI signals using the nine nuisance variables (three tissue signals and six motion parameters) along with their temporal derivatives. The Super Learner models included the mean model, generalized linear model, LASSO, ridge regression, elastic net, random forest, and multivariate adaptive regression splines \citep{SuperLearner_rpackage}. We conducted intra-subject processing on the average residual time series for regions of interest (ROIs) defined by Schaefer's 100-node cortical parcellation combined with 19 subcortical regions \citep{schaefer2018local}. We calculated pairwise Pearson correlations between the intra-subject processed residuals from all ROIs and performed Fisher Z-transformation to obtain the intra-subject processed derived FC estimates. We applied \texttt{ComBat} for site harmonization \citep{yu2018statistical} in which ``site" was a factor with three levels (NYU, KKI-8 channel, KKI-32 channel, see the Web Supplement Section C.1) with the following covariates: age, sex, full scale intelligence quotient, handedness, Autism Diagnostic Observation Schedule score, and mean framewise displacement.

In the inter-subject analysis, we included age, sex, full scale intelligence quotient, handedness, and Autism Diagnostic Observation Schedule score as common confounders and mean FD as a mediator. For our proposed method, inter-subject models $\hat{\pi}_A(a|\bw)$, $\hat{\pi}_A(a|m,\bw)$, and $\hat{b}_j(m,a,\bw)$ were also estimated using \texttt{SuperLearner} with the same set of learners as in the intra-subject processing step and 5-fold cross-fitting. We considered two classical approaches for comparison. The first approach used fMRI residuals from 36p linear nuisance regression and estimated treatment effects using a linear model without adjusting for mean FD (``36p + Linear"). The second approach also used linear 36p but with scrubbing, excluded subjects with $<5$ minutes of data remaining after scrubbing, and included mean FD in the inter-subject linear model (``36p Scrub + Linear M"). Notably, after applying scrubbing and subject exclusion, 96 children were removed from the analysis, leaving only 32 participants, of whom 6 were treated with stimulants. For multiple testing correction, the Benjamini–Hochberg procedure was applied to the results from the two classical methods to control the false discovery rate (FDR) at 0.1. For our proposed method, we controlled the exceedance rate of FDP at $P(\text{FDP} > 0.1) \le 0.1$.

\begin{figure}[t]
    \centering
    \includegraphics[width=\linewidth]{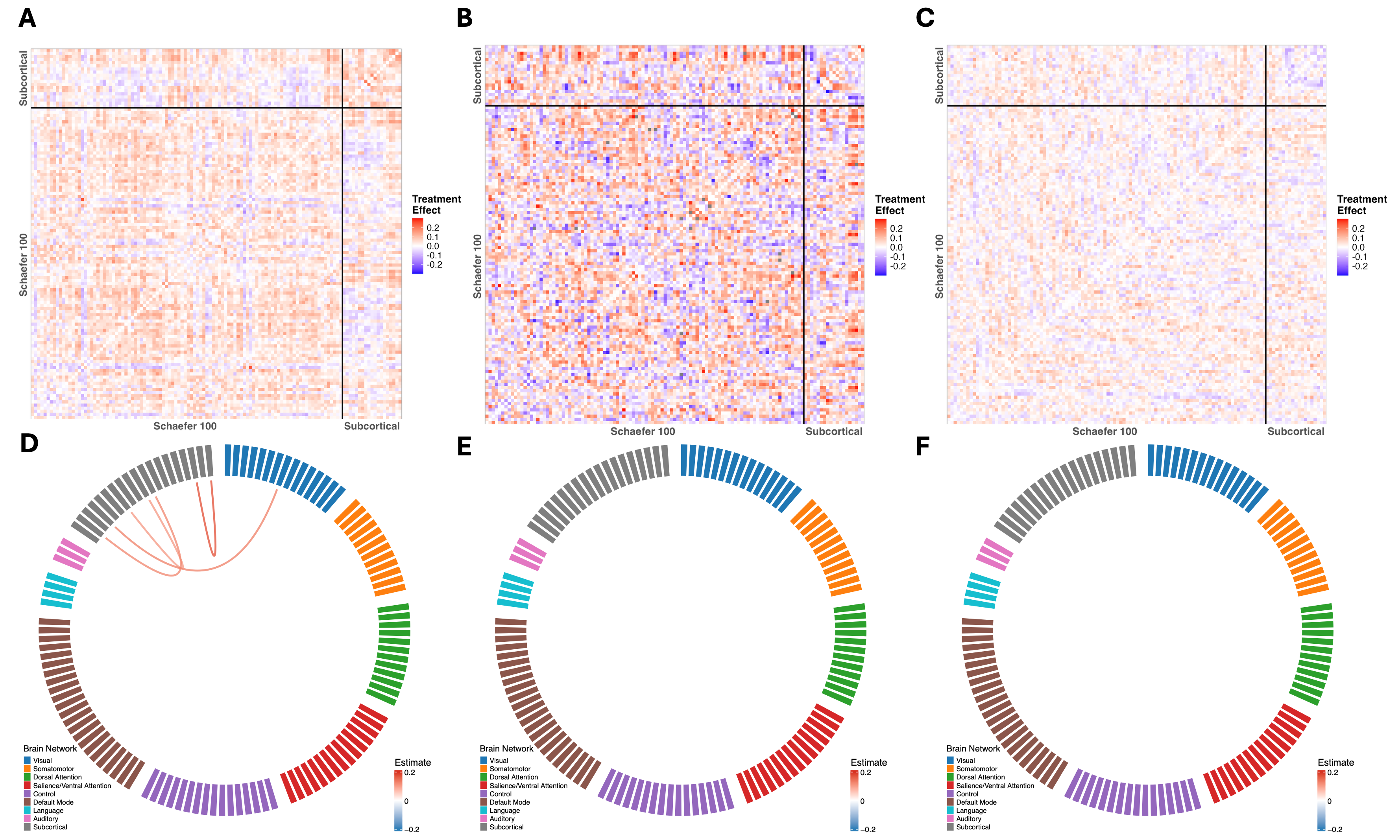}
    \caption{The stimulant treatment effects on FC between Scheafer 100 and 19 MNI subcortical brain regions using ``36p + Linear" (\textbf{A}), ``36p Scrub + Linear M" (\textbf{B}), and our proposed method (\textbf{C}). \textbf{D, E, F}: The circle plots of significant treatment effects estimations from three methods.}
    \label{fig:res_FC}
\end{figure}

\begin{figure}[t]
    \centering
    \includegraphics[width=\linewidth]{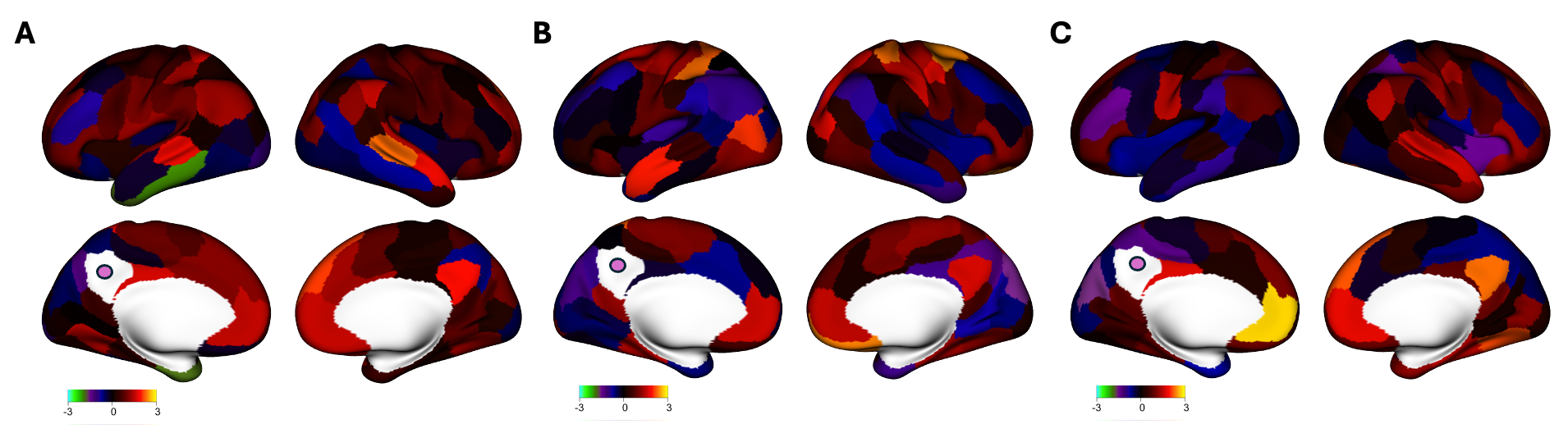}
    \caption{The brain map shows the Z-statistics for stimulant treatment effects on the FC between the seed region ``17networks\_LH\_DefaultA\_pCunPCC\_1" and all other regions in the Schaefer-100 parcellation, estimated using ``36p + Linear" (\textbf{A}), ``36p + Linear M" (\textbf{B}), and our proposed method (\textbf{C}).}
    \label{fig:res_seed3}
\end{figure}

\begin{table}
\centering
\small
\caption{The pairs of regions in which the uncorrected 95\% confidence interval does not include zero.}
\label{tab:nde-nie-ate}
\renewcommand{\arraystretch}{1.25}
\resizebox{\textwidth}{!}{%
\begin{tabular}{ccccc}
\toprule
\textbf{ROI 1} & \textbf{ROI 2} & \textbf{NDE (95\% CI)} & \textbf{NIE (95\% CI)} & \textbf{ATE (95\% CI)} \\
\midrule
\multirow{4}{*}{\centering pCunPCC1.L (Default.A)} 
& PFCm1.L (Default.A) & 0.134 (0.037, 0.230) & -0.048 (-0.102, 0.006) & 0.086 (0.004, 0.167) \\
& pCunPCC1.R (Default.A) & 0.133 (0.015, 0.250) & -0.043 (-0.093, 0.006) & 0.089 (-0.004, 0.182) \\
& PFCd1.R (Default.B) & 0.095 (0.006, 0.184) & -0.002 (-0.019, 0.015) & 0.093 (0.007, 0.180) \\
& ExStr3.R (Visual.A) & 0.080 (0.002, 0.157) & 0.001 (-0.023, 0.024) & 0.080 (0.008, 0.152) \\
\bottomrule
\end{tabular}
}
\end{table}

The estimated treatment effects of stimulant medications on FC are presented in \cref{fig:res_FC}. ``36p Scrub + Linear M" estimates are more variable, including larger treatment effects in many pairs of regions. This may be because 75\% of participants were removed in scrubbing step leaving only 6 treated participants, resulting in greater inaccuracies for some estimates. Out of the 7021 functional connectivity (FC) pairs analyzed, 988 (14.07\%), 370 (5.27\%), and 347 (4.94\%) estimates were marginally significant at the 0.05 level for the two classical approaches and our proposed method, respectively. Notably, there are far fewer p-values$<$0.05 when applying scrubbing with participant removal or our proposed method. This suggests that our approach greatly reduces motion artifacts relative to the ``36p + Linear" model. After applying multiple testing correction controlling the FDR at 0.1 for the classical methods and $P(\text{FDP} > 0.1) \le 0.1$ for our method, only four FC pairs from the ``36p + Linear" model remained significant, while no significant effects were identified for the ``36p Scrub + Linear M" approach or our proposed method. It is unclear whether the lack of significant findings is due to small-to-no long-term effects of stimulants on brain connectivity, heterogeneity in stimulant dosage and treatment duration within the sample, small sample size, and/or other factors. 

After this primary analysis, we conducted an exploratory post-hoc examination of the portion of the FC matrix corresponding to connections with the region ``17networks LH DefaultA pCunPCC 1", which is a hub of the posterior default mode network, and visualized the results with the \texttt{R} package \texttt{ciftiTools} \citep{pham2022ciftitools}. This seed region is prone to motion artifacts \citep{satterthwaite2012impact}. It is also thought to have reduced connectivity with anterior portions of the default mode network in autism \citep{DiMartino2013TheAutism}. The brain map of Z-statistics for stimulant treatment effects on the FC between this seed and all other regions in the Schaefer-100 parcellation is shown in \cref{fig:res_seed3}. In this post-hoc analysis, we examined uncorrected p-values.

The four regions with the most marginally significant NDE estimates, together with their ATE and NIE estimates using our proposed framework, are summarized in \cref{tab:nde-nie-ate}. The results indicate that indirect effects through motion tended to reduce the overall treatment effect, highlighting the need to account for motion in causal analysis. This is because motion was higher in the treated group, which attenuated the increase in FC from stimulant treatment in the ATE. Our method indicates an increase in the FC between the posterior default mode network seed region and the anterior default mode network, raising the possibility that stimulants may have normalizing effects in these regions. However, these findings were not statistically significant in our primary analysis adjusting for multiple comparisons and should be interpreted cautiously due to the limited sample size. Future studies with larger samples and less heterogeneity in treatment are needed to investigate these patterns.

\section{Discussion}
\label{sec:discussion}

In this study, we proposed a semiparametric derived outcome framework for causal inference with intra-subject processing followed by inter-subject analysis, a setting that is common in biomedical research. The intra-subject processing functions are unknown and must be estimated from intra-subject measurements. The subject-level outcomes are unobservable and are approximated using derived outcomes based on processed residuals. Our theoretical results show that valid inference with the proposed multiply robust estimator requires rate conditions on both intra-subject and inter-subject models. When the derived outcome has a cross-product structure, the intra-subject processing functions are allowed to converge at slower rates as the number of time points increases, relative to the number of subjects. This property allows flexible statistical and machine learning methods to be used for nuisance model estimation at both levels. We also adopt a multiple testing procedure to control the false discovery proportion when many effects are tested. In both simulation studies and the rs-fMRI application, we use Super Learner to estimate intra-subject and inter-subject nuisance models. We highlight that the proposed framework supports the use of adaptive learning methods while maintaining valid causal inference.

The proposed framework has several limitations and suggests directions for future work. First, the multiple robustness property holds at the inter-subject level but not at the intra-subject level. If the intra-subject processing function is not consistently estimated, the proposed estimator can be biased. This motivates the use of data-adaptive methods to reduce model misspecification. When more intra-subject data are available and stronger prior knowledge about the data structure exists, it may be possible to construct multiply robust derived outcomes to improve robustness. Second, constraints on the parameter space of subject-level outcomes can complicate modeling. In the rs-fMRI example, functional connectivity contains a positive definite correlation matrix. We apply a Fisher-Z transformation and model each connectivity value separately without considering positive definiteness. Dependence across the high-dimensional estimators is then addressed using a data-driven step-down multiple testing procedure based on estimated influence functions. Developing more powerful multiple testing methods that directly model dependence and/or leverage possible low rank structure remains an important topic for future research. Finally, the framework can be extended to study heterogeneity across groups or over time. Treatment effects within subgroups defined by covariates may provide additional insights. Longitudinal data may also allow analysis of how treatment effects evolve over time. As larger and more complex datasets become available, extending the framework to such settings will be an important direction for future work.

\setstretch{1.75}
\setlength{\bibsep}{2pt}
\bibliographystyle{agsm}
\bibliography{Bibliography-MM-MC}

@article{ciric2017benchmarking,
  title={Benchmarking of participant-level confound regression strategies for the control of motion artifact in studies of functional connectivity},
  author={Ciric, Rastko and Wolf, Daniel H and Power, Jonathan D and Roalf, David R and Baum, Graham L and Ruparel, Kosha and Shinohara, Russell T and Elliott, Mark A and Eickhoff, Simon B and Davatzikos, Christos and others},
  journal={Neuroimage},
  volume={154},
  pages={174--187},
  year={2017},
  publisher={Elsevier}
}

@article{cox1996afni,
  title={AFNI: software for analysis and visualization of functional magnetic resonance neuroimages},
  author={Cox, Robert W},
  journal={Computers and Biomedical research},
  volume={29},
  number={3},
  pages={162--173},
  year={1996},
  publisher={Elsevier}
}

@article{power2014methods,
  title={Methods to detect, characterize, and remove motion artifact in resting state fMRI},
  author={Power, Jonathan D and Mitra, Anish and Laumann, Timothy O and Snyder, Abraham Z and Schlaggar, Bradley L and Petersen, Steven E},
  journal={Neuroimage},
  volume={84},
  pages={320--341},
  year={2014},
  publisher={Elsevier}
}

@article{qiu2023unveiling,
  title={Unveiling the Unobservable: Causal Inference on Multiple Derived Outcomes},
  author={Qiu, Yumou and Sun, Jiarui and Zhou, Xiao-Hua},
  journal={Journal of the American Statistical Association},
  pages={1--12},
  year={2023},
  publisher={Taylor \& Francis}
}

@article{nebel2022accounting,
  title={Accounting for motion in resting-state fMRI: What part of the spectrum are we characterizing in autism spectrum disorder?},
  author={Nebel, Mary Beth and Lidstone, Daniel E and Wang, Liwei and others},
  journal={NeuroImage},
  volume={257},
  pages={119296},
  year={2022},
  publisher={Elsevier}
}

@article{marek2022reproducible,
  title={Reproducible brain-wide association studies require thousands of individuals},
  author={Marek, Scott and Tervo-Clemmens, Brenden and Calabro, Finnegan J and others},
  journal={Nature},
  volume={603},
  number={7902},
  pages={654--660},
  year={2022},
  publisher={Nature Publishing Group}
}

@article{breiman1996stacked,
  title={Stacked regressions},
  author={Breiman, Leo},
  journal={Machine learning},
  volume={24},
  pages={49--64},
  year={1996},
  publisher={Springer}
}

@article{van2007super,
  title={Super learner},
  author={Van der Laan, Mark J and Polley, Eric C and Hubbard, Alan E},
  journal={Statistical applications in genetics and molecular biology},
  volume={6},
  number={1},
  year={2007},
  publisher={De Gruyter}
}

@article{chernozhukov2018double,
  title={Double/debiased machine learning for treatment and structural parameters},
  author={Chernozhukov, Victor and Chetverikov, Denis and Demirer, Mert and others},
  journal={The Econometrics Journal},
  volume={21},
  number={1},
  pages={C1-C68},
  year={2018},
  publisher={Oxford University Press Oxford, UK}
}

@Manual{SuperLearner_rpackage,
    title = {SuperLearner: Super Learner Prediction},
    author = {Eric Polley and Erin LeDell and Chris Kennedy and Mark {van der Laan}},
    year = {2023},
    note = {R package version 2.0-28.1}
  }

@article{di2017enhancing,
  title={Enhancing studies of the connectome in autism using the autism brain imaging data exchange II},
  author={Di Martino, Adriana and O’connor, David and Chen, Bosi and Alaerts, Kaat and Anderson, Jeffrey S and Assaf and others},
  journal={Scientific data},
  volume={4},
  number={1},
  pages={1--15},
  year={2017},
  publisher={Nature Publishing Group}
}

@article{esteban2019fmriprep,
  title={fMRIPrep: a robust preprocessing pipeline for functional MRI},
  author={Esteban, Oscar and Markiewicz, Christopher J and Blair, Ross W and Moodie, Craig A and Isik, A and others},
  journal={Nature methods},
  volume={16},
  number={1},
  pages={111--116},
  year={2019},
  publisher={Nature Publishing Group US New York}
}

@article{schaefer2018local,
  title={Local-global parcellation of the human cerebral cortex from intrinsic functional connectivity MRI},
  author={Schaefer, Alexander and Kong, Ru and Gordon, Evan M and Laumann, Timothy O and Zuo, Xi-Nian and others},
  journal={Cerebral cortex},
  volume={28},
  number={9},
  pages={3095--3114},
  year={2018},
  publisher={Oxford University Press}
}

@article{yu2018statistical,
  title={Statistical harmonization corrects site effects in functional connectivity measurements from multi-site fMRI data},
  author={Yu, Meichen and Linn, Kristin A and Cook, Philip A and others},
  journal={Human brain mapping},
  volume={39},
  number={11},
  pages={4213--4227},
  year={2018},
  publisher={Wiley Online Library}
}

@article{pham2022ciftitools,
  title={ciftitools: A package for reading, writing, visualizing, and manipulating cifti files in r},
  author={Pham, Damon D and Muschelli, John and Mejia, Amanda F},
  journal={NeuroImage},
  volume={250},
  pages={118877},
  year={2022},
  publisher={Elsevier}
}

@article{fmriprep1,
    author = {Esteban, Oscar and Markiewicz, Christopher and Blair, Ross W and Moodie, Craig and Isik, Ayse Ilkay and Erramuzpe Aliaga, Asier and Kent, James and Goncalves, Mathias and DuPre, Elizabeth and Snyder, Madeleine and Oya, Hiroyuki and Ghosh, Satrajit and Wright, Jessey and Durnez, Joke and Poldrack, Russell and Gorgolewski, Krzysztof Jacek},
    title = {{fMRIPrep}: a robust preprocessing pipeline for functional {MRI}},
    year = {2018},
    doi = {10.1038/s41592-018-0235-4},
    journal = {Nature Methods}
}

@article{fmriprep2,
    author = {Esteban, Oscar and Blair, Ross and Markiewicz, Christopher J. and Berleant, Shoshana L. and Moodie, Craig and Ma, Feilong and Isik, Ayse Ilkay and Erramuzpe, Asier and Kent, James D. andGoncalves, Mathias and DuPre, Elizabeth and Sitek, Kevin R. and Gomez, Daniel E. P. and Lurie, Daniel J. and Ye, Zhifang and Poldrack, Russell A. and Gorgolewski, Krzysztof J.},
    title = {fMRIPrep},
    year = 2018,
    doi = {10.5281/zenodo.852659},
    publisher = {Zenodo},
    journal = {Software}
}

@article{nipype1,
    author = {Gorgolewski, K. and Burns, C. D. and Madison, C. and Clark, D. and Halchenko, Y. O. and Waskom, M. L. and Ghosh, S.},
    doi = {10.3389/fninf.2011.00013},
    journal = {Frontiers in Neuroinformatics},
    pages = 13,
    shorttitle = {Nipype},
    title = {Nipype: a flexible, lightweight and extensible neuroimaging data processing framework in Python},
    volume = 5,
    year = 2011
}

@article{nipype2,
    author = {Gorgolewski, Krzysztof J. and Esteban, Oscar and Markiewicz, Christopher J. and Ziegler, Erik and Ellis, David Gage and Notter, Michael Philipp and Jarecka, Dorota and Johnson, Hans and Burns, Christopher and Manhães-Savio, Alexandre and Hamalainen, Carlo and Yvernault, Benjamin and Salo, Taylor and Jordan, Kesshi and Goncalves, Mathias and Waskom, Michael and Clark, Daniel and Wong, Jason and Loney, Fred and Modat, Marc and Dewey, Blake E and Madison, Cindee and Visconti di Oleggio Castello, Matteo and Clark, Michael G. and Dayan, Michael and Clark, Dav and Keshavan, Anisha and Pinsard, Basile and Gramfort, Alexandre and Berleant, Shoshana and Nielson, Dylan M. and Bougacha, Salma and Varoquaux, Gael and Cipollini, Ben and Markello, Ross and Rokem, Ariel and Moloney, Brendan and Halchenko, Yaroslav O. and Wassermann , Demian and Hanke, Michael and Horea, Christian and Kaczmarzyk, Jakub and Gilles de Hollander and DuPre, Elizabeth and Gillman, Ashley and Mordom, David and Buchanan, Colin and Tungaraza, Rosalia and Pauli, Wolfgang M. and Iqbal, Shariq and Sikka, Sharad and Mancini, Matteo and Schwartz, Yannick and Malone, Ian B. and Dubois, Mathieu and Frohlich, Caroline and Welch, David and Forbes, Jessica and Kent, James and Watanabe, Aimi and Cumba, Chad and Huntenburg, Julia M. and Kastman, Erik and Nichols, B. Nolan and Eshaghi, Arman and Ginsburg, Daniel and Schaefer, Alexander and Acland, Benjamin and Giavasis, Steven and Kleesiek, Jens and Erickson, Drew and Küttner, René and Haselgrove, Christian and Correa, Carlos and Ghayoor, Ali and Liem, Franz and Millman, Jarrod and Haehn, Daniel and Lai, Jeff and Zhou, Dale and Blair, Ross and Glatard, Tristan and Renfro, Mandy and Liu, Siqi and Kahn, Ari E. and Pérez-García, Fernando and Triplett, William and Lampe, Leonie and Stadler, Jörg and Kong, Xiang-Zhen and Hallquist, Michael and Chetverikov, Andrey and Salvatore, John and Park, Anne and Poldrack, Russell and Craddock, R. Cameron and Inati, Souheil and Hinds, Oliver and Cooper, Gavin and Perkins, L. Nathan and Marina, Ana and Mattfeld, Aaron and Noel, Maxime and Lukas Snoek and Matsubara, K and Cheung, Brian and Rothmei, Simon and Urchs, Sebastian and Durnez, Joke and Mertz, Fred and Geisler, Daniel and Floren, Andrew and Gerhard, Stephan and Sharp, Paul and Molina-Romero, Miguel and Weinstein, Alejandro and Broderick, William and Saase, Victor and Andberg, Sami Kristian and Harms, Robbert and Schlamp, Kai and Arias, Jaime and Papadopoulos Orfanos, Dimitri and Tarbert, Claire and Tambini, Arielle and De La Vega, Alejandro and Nickson, Thomas and Brett, Matthew and Falkiewicz, Marcel and Podranski, Kornelius and Linkersdörfer, Janosch and Flandin, Guillaume and Ort, Eduard and Shachnev, Dmitry and McNamee, Daniel and Davison, Andrew and Varada, Jan and Schwabacher, Isaac and Pellman, John and Perez-Guevara, Martin and Khanuja, Ranjeet and Pannetier, Nicolas and McDermottroe, Conor and Ghosh, Satrajit},
    title = {Nipype},
    year = 2018,
    doi = {10.5281/zenodo.596855},
    publisher = {Zenodo},
    journal = {Software}
}

@article{n4,
    author = {Tustison, N. J. and Avants, B. B. and Cook, P. A. and Zheng, Y. and Egan, A. and Yushkevich, P. A. and Gee, J. C.},
    doi = {10.1109/TMI.2010.2046908},
    issn = {0278-0062},
    journal = {IEEE Transactions on Medical Imaging},
    number = 6,
    pages = {1310-1320},
    shorttitle = {N4ITK},
    title = {N4ITK: Improved N3 Bias Correction},
    volume = 29,
    year = 2010
}

@article{fs_reconall,
    author = {Dale, Anders M. and Fischl, Bruce and Sereno, Martin I.},
    doi = {10.1006/nimg.1998.0395},
    issn = {1053-8119},
    journal = {NeuroImage},
    number = 2,
    pages = {179-194},
    shorttitle = {Cortical Surface-Based Analysis},
    title = {Cortical Surface-Based Analysis: I. Segmentation and Surface Reconstruction},
    url = {http://www.sciencedirect.com/science/article/pii/S1053811998903950},
    volume = 9,
    year = 1999
}

@article{mindboggle,
    author = {Klein, Arno and Ghosh, Satrajit S. and Bao, Forrest S. and Giard, Joachim and Häme, Yrjö and Stavsky, Eliezer and Lee, Noah and Rossa, Brian and Reuter, Martin and Neto, Elias Chaibub and Keshavan, Anisha},
    doi = {10.1371/journal.pcbi.1005350},
    issn = {1553-7358},
    journal = {PLOS Computational Biology},
    number = 2,
    pages = {e1005350},
    title = {Mindboggling morphometry of human brains},
    url = {http://journals.plos.org/ploscompbiol/article?id=10.1371/journal.pcbi.1005350},
    volume = 13,
    year = 2017
}

@article{mni152nlin2009casym,
    title = {Unbiased nonlinear average age-appropriate brain templates from birth to adulthood},
    author = {Fonov, VS and Evans, AC and McKinstry, RC and Almli, CR and Collins, DL},
    doi = {10.1016/S1053-8119(09)70884-5},
    journal = {NeuroImage},
    pages = {S102},
    volume = {47, Supplement 1},
    year = 2009
}

@article{mni152nlin6asym,
    author = {Evans, AC and Janke, AL and Collins, DL and Baillet, S},
    title = {Brain templates and atlases},
    doi = {10.1016/j.neuroimage.2012.01.024},
    journal = {NeuroImage},
    volume = {62},
    number = {2},
    pages = {911--922},
    year = 2012
}

@article{ants,
    author = {Avants, B.B. and Epstein, C.L. and Grossman, M. and Gee, J.C.},
    doi = {10.1016/j.media.2007.06.004},
    issn = {1361-8415},
    journal = {Medical Image Analysis},
    number = 1,
    pages = {26-41},
    shorttitle = {Symmetric diffeomorphic image registration with cross-correlation},
    title = {Symmetric diffeomorphic image registration with cross-correlation: Evaluating automated labeling of elderly and neurodegenerative brain},
    url = {http://www.sciencedirect.com/science/article/pii/S1361841507000606},
    volume = 12,
    year = 2008
}

@article{fsl_fast,
    author = {Zhang, Y. and Brady, M. and Smith, S.},
    doi = {10.1109/42.906424},
    issn = {0278-0062},
    journal = {IEEE Transactions on Medical Imaging},
    number = 1,
    pages = {45-57},
    title = {Segmentation of brain {MR} images through a hidden Markov random field model and the expectation-maximization algorithm},
    volume = 20,
    year = 2001
}

@article{mcflirt,
    author = {Jenkinson, Mark and Bannister, Peter and Brady, Michael and Smith, Stephen},
    doi = {10.1006/nimg.2002.1132},
    issn = {1053-8119},
    journal = {NeuroImage},
    number = 2,
    pages = {825-841},
    title = {Improved Optimization for the Robust and Accurate Linear Registration and Motion Correction of Brain Images},
    url = {http://www.sciencedirect.com/science/article/pii/S1053811902911328},
    volume = 17,
    year = 2002
}

@article{bbr,
    author = {Greve, Douglas N and Fischl, Bruce},
    doi = {10.1016/j.neuroimage.2009.06.060},
    issn = {1095-9572},
    journal = {NeuroImage},
    number = 1,
    pages = {63-72},
    title = {Accurate and robust brain image alignment using boundary-based registration},
    volume = 48,
    year = 2009
}

@article{power_fd_dvars,
    author = {Power, Jonathan D. and Mitra, Anish and Laumann, Timothy O. and Snyder, Abraham Z. and Schlaggar, Bradley L. and Petersen, Steven E.},
    doi = {10.1016/j.neuroimage.2013.08.048},
    issn = {1053-8119},
    journal = {NeuroImage},
    number = {Supplement C},
    pages = {320-341},
    title = {Methods to detect, characterize, and remove motion artifact in resting state fMRI},
    url = {http://www.sciencedirect.com/science/article/pii/S1053811913009117},
    volume = 84,
    year = 2014
}

@article{nilearn,
    author = {Abraham, Alexandre and Pedregosa, Fabian and Eickenberg, Michael and Gervais, Philippe and Mueller, Andreas and Kossaifi, Jean and Gramfort, Alexandre and Thirion, Bertrand and Varoquaux, Gael},
    doi = {10.3389/fninf.2014.00014},
    issn = {1662-5196},
    journal = {Frontiers in Neuroinformatics},
    language = {English},
    title = {Machine learning for neuroimaging with scikit-learn},
    url = {https://www.frontiersin.org/articles/10.3389/fninf.2014.00014/full},
    volume = 8,
    year = 2014
}

@article{lanczos,
    author = {Lanczos, C.},
    doi = {10.1137/0701007},
    issn = {0887-459X},
    journal = {Journal of the Society for Industrial and Applied Mathematics Series B Numerical Analysis},
    number = 1,
    pages = {76-85},
    title = {Evaluation of Noisy Data},
    url = {http://epubs.siam.org/doi/10.1137/0701007},
    volume = 1,
    year = 1964
}

@article{hcppipelines,
    author = {Glasser, Matthew F. and Sotiropoulos, Stamatios N. and Wilson, J. Anthony and Coalson, Timothy S. and Fischl, Bruce and Andersson, Jesper L. and Xu, Junqian and Jbabdi, Saad and Webster, Matthew and Polimeni, Jonathan R. and Van Essen, David C. and Jenkinson, Mark},
    doi = {10.1016/j.neuroimage.2013.04.127},
    issn = {1053-8119},
    journal = {NeuroImage},
    pages = {105-124},
    series = {Mapping the Connectome},
    title = {The minimal preprocessing pipelines for the Human Connectome Project},
    url = {http://www.sciencedirect.com/science/article/pii/S1053811913005053},
    volume = 80,
    year = 2013
}

@article{afni,
    author = {Cox, Robert W. and Hyde, James S.},
    doi = {10.1002/(SICI)1099-1492(199706/08)10:4/5<171::AID-NBM453>3.0.CO;2-L},
    journal = {NMR in Biomedicine},
    number = {4-5},
    pages = {171-178},
    title = {Software tools for analysis and visualization of fMRI data},
    volume = 10,
    year = 1997
}

@book{van2011targeted,
  title={Targeted learning: causal inference for observational and experimental data},
  author={Van der Laan, Mark J and Rose, Sherri and others},
  volume={4},
  year={2011},
  publisher={Springer}
}

@article{cosgrove2022limits,
  title={Limits to the generalizability of resting-state functional magnetic resonance imaging studies of youth: An examination of ABCD Study{\textregistered} baseline data},
  author={Cosgrove, Kelly T and McDermott, Timothy J and White, Evan J and others},
  journal={Brain imaging and behavior},
  volume={16},
  number={4},
  pages={1919--1925},
  year={2022},
  publisher={Springer}
}

@article{yan2013comprehensive,
  title={A comprehensive assessment of regional variation in the impact of head micromovements on functional connectomics},
  author={Yan, Chao-Gan and Cheung, Brian and Kelly, Clare and Colcombe, Stan and Craddock, R Cameron and Di Martino, Adriana and Li, Qingyang and Zuo, Xi-Nian and Castellanos, F Xavier and Milham, Michael P},
  journal={Neuroimage},
  volume={76},
  pages={183--201},
  year={2013},
  publisher={Elsevier}
}

@article{cao2009improving,
  title={Improving efficiency and robustness of the doubly robust estimator for a population mean with incomplete data},
  author={Cao, Weihua and Tsiatis, Anastasios A and Davidian, Marie},
  journal={Biometrika},
  volume={96},
  number={3},
  pages={723--734},
  year={2009},
  publisher={Oxford University Press}
}

@article{chernozhukov2013gaussian,
  title={Gaussian approximations and multiplier bootstrap for maxima of sums of high-dimensional random vectors},
  author={Chernozhukov, Victor and Chetverikov, Denis and Kato, Kengo},
  journal={Annals of Statistics},
  volume={41},
  number={6},
 pages={2786--2819},
  year={2013}
}

@article{du2025causal,
  title={Causal Inference for Genomic Data with Multiple Heterogeneous Outcomes},
  author={Du, Jin-Hong and Zeng, Zhenghao and Kennedy, Edward H. and Wasserman, Larry and Roeder, Kathryn},
  journal={Journal of the American Statistical Association},
  volume={120},
  number={552},
  pages={2484--2497},
  year={2025},
  doi={10.1080/01621459.2025.2468014},
  publisher={Taylor \& Francis}
}

@article{genovese2006exceedance,
  title={Exceedance control of the false discovery proportion},
  author={Genovese, Christopher R and Wasserman, Larry},
  journal={Journal of the American Statistical Association},
  volume={101},
  number={476},
  pages={1408--1417},
  year={2006},
  publisher={Taylor \& Francis}
}

@book{tsiatis2006semiparametric,
  title={Semiparametric theory and missing data},
  author={Tsiatis, Anastasios A},
  year={2006},
  publisher={Springer}
}

@article{DiMartino2013TheAutism,
    title = {{The autism brain imaging data exchange: towards a large-scale evaluation of the intrinsic brain architecture in autism}},
    year = {2013},
    journal = {Molecular Psychiatry 2014 19:6},
    author = {Di Martino, A. and Al, Et. and Milham, M. P.},
    number = {6},
    month = {6},
    pages = {659--667},
    volume = {19},
    publisher = {Nature Publishing Group},
    doi = {10.1038/mp.2013.78},
    issn = {1476-5578},
    pmid = {23774715}
}

@article{tchetgen2012semiparametric,
  title={Semiparametric theory for causal mediation analysis: efficiency bounds, multiple robustness, and sensitivity analysis},
  author={Tchetgen, Eric J Tchetgen and Shpitser, Ilya},
  journal={Annals of statistics},
  volume={40},
  number={3},
  pages={1816},
  year={2012}
}

@book{vanderweele2015explanation,
  title={Explanation in causal inference: methods for mediation and interaction},
  author={VanderWeele, Tyler},
  year={2015},
  publisher={Oxford University Press}
}

@book{van2000asymptotic,
  title={Asymptotic statistics},
  author={Van der Vaart, Aad W},
  volume={3},
  year={2000},
  publisher={Cambridge university press}
}

@article{kennedy2020sharp,
  title={Sharp instruments for classifying compliers and generalizing causal effects},
  author={Kennedy, Edward H and Balakrishnan, Sivaraman and G’sell, Max},
  year={2020}
}

@article{belloni2018high,
  title={High-dimensional econometrics and regularized GMM},
  author={Belloni, Alexandre and Chernozhukov, Victor and Chetverikov, Denis and Hansen, Christian and Kato, Kengo},
  journal={arXiv preprint arXiv:1806.01888},
  year={2018}
}

@book{imbens2015causal,
  title={Causal inference in statistics, social, and biomedical sciences},
  author={Imbens, Guido W and Rubin, Donald B},
  year={2015},
  publisher={Cambridge university press}
}

@article{satterthwaite2012impact,
  title={Impact of in-scanner head motion on multiple measures of functional connectivity: relevance for studies of neurodevelopment in youth},
  author={Satterthwaite, Theodore D and Wolf, Daniel H and Loughead, James and Ruparel, Kosha and Elliott, Mark A and Hakonarson, Hakon and Gur, Ruben C and Gur, Raquel E},
  journal={Neuroimage},
  volume={60},
  number={1},
  pages={623--632},
  year={2012},
  publisher={Elsevier}
}

@inproceedings{di2016new,
  title={A new regression-based method for the eye blinks artifacts correction in the EEG signal, without using any EOG channel},
  author={Di Flumeri, Gianluca and Aric{\`o}, Pietro and Borghini, Gianluca and Colosimo, Alfredo and Babiloni, Fabio},
  booktitle={2016 38th Annual International Conference of the IEEE Engineering in Medicine and Biology Society (EMBC)},
  pages={3187--3190},
  year={2016},
  organization={IEEE}
}

@article{stringer2019computational,
  title={Computational processing of neural recordings from calcium imaging data},
  author={Stringer, Carsen and Pachitariu, Marius},
  journal={Current opinion in neurobiology},
  volume={55},
  pages={22--31},
  year={2019},
  publisher={Elsevier}
}

@article{robins1994estimation,
  title={Estimation of regression coefficients when some regressors are not always observed},
  author={Robins, James M and Rotnitzky, Andrea and Zhao, Lue Ping},
  journal={Journal of the American statistical Association},
  volume={89},
  number={427},
  pages={846--866},
  year={1994},
  publisher={Taylor \& Francis}
}

@article{scharfstein1999adjusting,
  title={Adjusting for nonignorable drop-out using semiparametric nonresponse models},
  author={Scharfstein, Daniel O and Rotnitzky, Andrea and Robins, James M},
  journal={Journal of the American Statistical Association},
  volume={94},
  number={448},
  pages={1096--1120},
  year={1999},
  publisher={Taylor \& Francis}
}

@book{laan2003unified,
  title={Unified methods for censored longitudinal data and causality},
  author={Laan, Mark J and Robins, James M},
  year={2003},
  publisher={Springer}
}

@article{zhang2024motion,
  title={Motion-invariant variational autoencoding of brain structural connectomes},
  author={Zhang, Yizi and Liu, Meimei and Zhang, Zhengwu and Dunson, David},
  journal={Imaging Neuroscience},
  volume={2},
  pages={1--27},
  year={2024},
  publisher={MIT Press 255 Main Street, 9th Floor, Cambridge, Massachusetts 02142, USA~…}
}

@article{chaudhary2022fast,
  title={Fast, efficient, and accurate neuro-imaging denoising via supervised deep learning},
  author={Chaudhary, Shivesh and Moon, Sihoon and Lu, Hang},
  journal={Nature communications},
  volume={13},
  number={1},
  pages={5165},
  year={2022},
  publisher={Nature Publishing Group UK London}
}

@article{manzano2024denoising,
  title={Denoising diffusion MRI: Considerations and implications for analysis},
  author={Manzano Patron, Jose Pedro and Moeller, Steen and Andersson, Jesper LR and Ugurbil, Kamil and Yacoub, Essa and Sotiropoulos, Stamatios N},
  journal={Imaging Neuroscience},
  volume={2},
  pages={1--29},
  year={2024},
  publisher={MIT Press One Broadway, 12th Floor, Cambridge, Massachusetts 02142, USA~…}
}

\newpage
\appendix

\noindent {\bf \LARGE Web Supplement} 

\section{Derivations and Proofs}

\subsection{Identification with  intra-subject processing}

\textbf{Proof of Theorem 1.} 
Recall that $\bm{\psi}(a,a') = \tE[\bY(a, a')] = \tE[\bY(a, M(a'))]$. Define $\hbY(a, a') = \hbY(a, M(a'))$ and $\bm{\Delta}(a, a') = \bm{\Delta}(a, M(a'))$. Assumptions 1-4 (consistency, positivity, conditional randomization, and cross-world independence) are imposed on the full intra-subject measurements $\bX(a,m)$ and $\bH(a,m)$. By Definition 1 and Definition 2, these assumptions carry over to the induced subject-level outcomes $\bY(a,m)$ and derived outcomes $\hbY(a,m)$. In particular, the following conditional independence relations hold: 
\begin{align*}
&(i) \{ \bY(a,m),\hbY(a,m) \} \perp M(a') \mid \bW;\\
&(ii) M(a) \perp A \mid \bW;\\
&(iii) \{ \bY(a,m),\hbY(a,m) \} \perp A \mid \bW; \\
&(iv) \{ \bY(m),\hbY(m) \} \perp M \mid A,\bW.
\end{align*}
Using these relations, we have
\begin{align*}
&\tE[\tE[\hbY \mid A = a, M, \bW] \mid A=a',\bW]\\ 
&= \int_{m} \tE[\hbY | A=a, M=m, \bW] \times p[M=m | A=a', \bW] dm\\
&\overset{(Con)}{=} \int_{m} \tE[\hbY(m) | A=a, M=m, \bW] \times p[M=m | A=a', \bW] dm\\
&\overset{(Con,iv)}{=} \int_{m} \tE[\hbY(a,m) | A=a, \bW] \times p[M(a')=m | A=a', \bW] dm\\
&\overset{(ii,iii)}{=} \int_{m} \tE[\hbY(a,m) | \bW] \times p[M(a')=m | \bW] dm\\
&\overset{(i)}{=} \int_{m} \tE[\hbY(a,m) | M(a') = m, \bW] \times p[M(a')=m | \bW] dm\\
&= \tE[\tE[\hbY(a,a') \mid \bW] \mid \bW].
\end{align*}
By the definition of the discrepancy between the derived and true outcomes in Equation (6), we have
$$
\bm{\Delta}_T(a,m) + \bY(a,m) = \tE[\hbY(a,m) \mid \bU_D(a), \bU_I(a), \bW].
$$
It follows that
\begin{align*}
&\tE[\tE[\hbY \mid A = a, M, \bW] \mid A=a',\bW]\\ 
&= \tE[\tE[\hbY(a,a') \mid \bW] \mid \bW]\\ 
&= \tE[ \tE[ \tE[\hbY(a,a') \mid \bU_D(a), \bU_I(a), \bW] \mid \bW] \mid \bW] \\
&= \tE[\tE[\bY(a,a') + \bm{\Delta}_T(a,a') \mid \bW] \mid \bW]\\ 
&= \tE[\tE[\bY(a,a') \mid \bW] \mid \bW] + \tE[\tE[\bm{\Delta}_T(a,a') \mid \bW] \mid \bW].
\end{align*}
By the asymptotic unbiasedness in Assumption 5, as $T \rightarrow \infty$, we have
\begin{align*}
&\tE[\tE[\tE[\hbY \mid A = a, M, \bW] \mid A = a', \bW]] \\
&= \tE[\tE[\tE[\bY(a,a') \mid \bW] \mid \bW]] + \tE[\tE[\tE[\bm{\Delta}_T(a,a') \mid \bW] \mid \bW]] \\
&= \tE[\bY(a,a') ] + \tE[\bm{\Delta}_T(a,a')] \\
&= \tE[\bY(a,a') ] + o(1).
\end{align*}
Therefore, we now have $\bm{\psi}(a,a') = \tE[\bY(a,a')] = \tE[\tE[\tE[\hbY \mid A = a, M, \bW] \mid A = a', \bW]]$ as $T \rightarrow \infty$. This completes the proof. \qed

\subsection{Decomposition of estimation error}

\textbf{Proof of Equation (8).} 
We define $\hat{\phi}(\tP)$ as the influence function constructed using the derived outcome $\hY$, satisfying $\tP[\hat{\phi}(\tP)] = 0$ and $\tP[\hat{\phi}^2(\tP)] < \infty$. We are considering an estimator $\hat{\psi}$ satisfies the von Mises Taylor expansion
\begin{equation*}
\begin{aligned}
\hat{\psi}(\tP_1) - \hat{\psi}(\tP_2) &= \int \hat{\phi}(\tP_1) d(\tP_1 - \tP_2) + \hat{R}_2(\tP_1,\tP_2)\\
&= - \int \hat{\phi}(\tP_1) d \tP_2 + \hat{R}_2(\tP_1,\tP_2),
\end{aligned}
\end{equation*}
where $\tP_1, \tP_2 \in \cP$, and $\hat{R}_2(\tP_1,\tP_2)$ is the derived outcome based second order remainder term. The plug-in estimator $\hat{\psi}(\hat{\tP})$ uses $\hat{\tP}$ to estimate $\tP$, and typically has first order bias:
\begin{equation*}
\hat{\psi}(\hat{\tP}) - \hat{\psi}(\tP) = - \int \hat{\phi}(\hat{\tP}) d \tP + \hat{R}_2(\hat{\tP},\tP).
\end{equation*}
It suggests a bias-correction procedure. The first order bias can be estimated using the empirical distribution $\tP_n$,
\begin{equation*}
\hat{\text{Bias}} = - \int \hat{\phi}(\hat{\tP}) d \tP_n.
\end{equation*}
The one-step estimator, which is a corrected plug-in estimator, is of the form
\begin{equation*}
\begin{aligned}
\hat{\psi}_n(\tP_n,\hat{\tP}) &= \hat{\psi}(\hat{\tP}) - \hat{\text{Bias}}\\
&= \hat{\psi}(\hat{\tP}) + \int \hat{\phi}(\hat{\tP}) d \tP_n.
\end{aligned}
\end{equation*}
The asymptotic behavior of the one-step estimator follows that
\begin{equation*}
\begin{aligned}
\hat{\psi}_n(\tP_n,\hat{\tP}) - \hat{\psi}(\tP) &= \left\{ \hat{\psi}(\hat{\tP}) + \int \hat{\phi}(\hat{\tP}) d \tP_n \right\} - \hat{\psi}(\tP)\\
&= \left\{ \hat{\psi}(\hat{\tP}) - \hat{\psi}(\tP) \right\} + \int \hat{\phi}(\hat{\tP}) d \tP_n \\
&= \left\{ - \int \hat{\phi}(\hat{\tP}) d \tP + \hat{R}_2(\hat{\tP},\tP) \right\} + \int \hat{\phi}(\hat{\tP}) d \tP_n \\
&= \int \hat{\phi}(\hat{\tP}) d \left\{ \tP_n - \tP \right\} + \hat{R}_2(\hat{\tP},\tP) \\
&= \int \hat{\phi}(\tP) d \tP_n  + \int \left\{ \hat{\phi}(\hat{\tP}) - \hat{\phi}(\tP) \right\} d \left\{ \tP_n - \tP \right\} + \hat{R}_2(\hat{\tP},\tP) \\
&= \tP_n [\hat{\phi}(\tP)]  + (\tP_n - \tP) \{ \hat{\phi}(\hat{\tP}) - \hat{\phi}(\tP) \} + \hat{R}_2(\hat{\tP},\tP).
\end{aligned}
\end{equation*}
We apply this decomposition to the $j$th component of the derived outcome–based target parameter $\bm{\hat{\psi}}(a,a') = \tE[\hbY(a, a')]$ and its one-step estimator $\bm{\hat{\psi}}_{n,j}(a,a')$. By the definition of the discrepancy between the derived and true outcomes in Equation (6), we have
\begin{equation*}
\begin{aligned}
\hat{\psi}_{n,j}(a,a') - \psi_j(a,a') &= \hat{\psi}_{n,j}(a,a') - \hat{\psi}_{j}(a,a') + \hat{\psi}_{j}(a,a') - \psi_j(a,a')\\
&= \hat{\psi}_{n,j}(a,a') - \hat{\psi}_{j}(a,a') + \tE[\hY_j(a, M(a'))] - \tE[Y_j(a, M(a'))]\\
&= \tP_n [\hat{\phi}_j(\tP)] + (\tP_n - \tP)\{\hat{\phi}_j(\hat{\tP}) - \hat{\phi}_j(\tP)\} + \hat{R}_2(\hat{\tP},\tP) + \tE[ \Delta_{jT}(a,a') ].
\end{aligned}
\end{equation*}
This completes the proof. \qed

\subsection{Multiply robust estimation}

\textbf{Proof of Lemma 1.} 
The $\Delta_{jT}(a,m)$ can be decomposed into two parts as shown in Equation (9):
\begin{align*}
\Delta_{jT}(a,m) &= \tE[\hat{Y}_{j}(a,m) \mid \bU_D(a),\bU_I(a), \bW] -  Y_{j}(a,m)\\
&= \tE[\hat{Y}_{j}(a,m) \mid \bU_D(a),\bU_I(a),\bW] - \tE[\Tilde{Y}_{j}(a,m) \mid \bU_D(a),\bU_I(a), \bW]\\
&+ \tE[\Tilde{Y}_{j}(a,m) \mid \bU_D(a),\bU_I(a),\bW] -  Y_{j}(a,m),
\end{align*}
where the outcome with intra-subject processing $Y_j(a,m)$, the derived outcome with intra-subject processing $\hat{Y}_j(a,m)$, and the intermediate functional $\Tilde{Y}_j(a,m)$ are in cross-product type as defined in Definition 3. Specifically, define $\tP_T [\cdot]$ as the empirical average over intra-subject measurements, we have
\begin{equation*}
\begin{aligned}
Y_j(a,m) &= g \left[ \tE\left\{ \bX_{tv}(a,m) - f_v(\bH_t(a,m)) \right\}
\left\{ \bX_{tv'}(a,m) - f_{v'}(\bH_t(a,m)) \right\}
\mid \bU_D(a), \bU_I(a), \bW \right],\\
\hat{Y}_j(a,m) &= \hat{g} \{ \tP_T[ \bX_{v}(a,m) - \hat{f}_v(\bH(a,m))] [\bX_{v'}(a,m) - \hat{f}_{v'}(\bH(a,m))] \},\\
\Tilde{Y}_j(a,m) &= \hat{g} \{ \tP_T[ \bX_{v}(a,m) - f_v(\bH(a,m))] [\bX_{v'}(a,m) - f_{v'}(\bH(a,m))] \}.
\end{aligned}
\end{equation*}

Within each subject, let $\eta_{vt}(a,m) = X_{vt}(a,m) - f_{v}(\bH_{t}(a,m))$ denote the true intra-subject residuals, and $\hat{\eta}_{vt}(a,m) = X_{vt}(a,m) - \hat{f}_{v}(\bH_{t}(a,m))$ denote the estimated residuals, for $t = 1, \ldots,T$. Therefore, their contrast is $\hat{\eta}_{vt}(a,m) - \eta_{vt}(a,m) = f_{v}(\bH_{t}(a,m)) - \hat{f}_{v}(\bH_{t}(a,m))$. The estimated residual converges to the truth as the estimated intra-subject function converges. Without loss of generality, assume the residuals are centered in a sense that $\tP_T [\bm{\hat{\eta}}_{v}(a,m)]=0$ and $\tP_T [\bm{\eta}_{v}(a,m)]=0$.  Due to existence of possible temporal dependence and the use of data-adaptive models, we denote the rate of intra-subject processing satisfies $ \| \hat{f}_v - f_v \|_2 = \hat{\eta}_{vt} - \eta_{vt} = \cO(T^{-\alpha_0})$, for some $\alpha_0 \in (0, 1/2]$. Note that
\begin{align*}
&\tP_T \{ [ \bX_{v}(a,m) - \hat{f}_v(\bH(a,m))] [\bX_{v'}(a,m) - \hat{f}_{v'}(\bH(a,m))] \}\\
&= \tP_T [\hat{\eta}_{v}(a,m) \hat{\eta}_{v'}(a,m)] \\
&= \frac{1}{T} \sum_{t=1}^T \hat{\eta}_{vt}(a,m) \hat{\eta}_{v't}(a,m)\\
&= \frac{1}{T} \sum_{t=1}^T [\eta_{vt}(a,m) + O(T^{-\alpha})] [\eta_{v't}(a,m) + O(T^{-\alpha})]\\
&= \frac{1}{T} \sum_{t=1}^T \eta_{vt}(a,m) \eta_{v't}(a,m) + \frac{1}{T} \sum_{t=1}^T \eta_{vt}(a,m) O(T^{-\alpha}) \\
&+ \frac{1}{T} \sum_{t=1}^T \eta_{v't}(a,m) O(T^{-\alpha}) + \frac{1}{T} \sum_{t=1}^T O(T^{-\alpha}) O(T^{-\alpha})\\
&= \tP_T [\eta_{vt}(a,m) \eta_{v't}(a,m)] + O(T^{-2\alpha}) \\
&= \tP_T[ \bX_{v}(a,m) - f_v(\bH(a,m))] [\bX_{v'}(a,m) - f_{v'}(\bH(a,m))] + O(T^{-2\alpha}).
\end{align*}
By applying the continuous mapping theorem to the empirical outcome function $\hat{g}$, and  taking expectation on both sides, we have 
\begin{equation*}
\begin{aligned}
& \tE[ \hat{g} \{ \tP_T (\hat{\eta}_{v}(a,m) \hat{\eta}_{v'}(a,m)) \} \mid \bU_D(a),\bU_I(a),\bW]\\
&= \tE[ \hat{g} \{ \tP_T ( \eta_{v}(a,m) \eta_{v'}(a,m) ) \} \mid \bU_D(a),\bU_I(a),\bW] + O(T^{-2\alpha}).
\end{aligned}
\end{equation*}
Then by Assumption 7, as $n,T \rightarrow \infty$, we have 
\begin{align*}
&\sqrt{n} \big[ \tE \big\{ \hat{g} \Big(  \tP_T \big(\hat{\eta}_{v}(a,m) \hat{\eta}_{v'}(a,m) \big) \Big) \mid \bU_D(a),\bU_I(a),\bW \big\} \\
&- \tE \big\{ \hat{g} \Big( \tP_T \big( \eta_{v}(a,m) \eta_{v'}(a,m) \big) \Big) \mid \bU_D(a),\bU_I(a),\bW \big\} \big] \\
&= \sqrt{n} \big[ \tE \{\hat{Y}_{j}(a,m) \mid \bU_D(a),\bU_I(a),\bW \} - \tE \{ \Tilde{Y}_{j}(a,m) \mid \bU_D(a),\bU_I(a), \bW \} \big]\\
&\rightarrow 0.
\end{align*}
By applying the continuous mapping theorem to the empirical outcome function $\hat{g}$, we have 
\begin{equation*}
\begin{aligned}
& \sqrt{n}\{ \tE[ \tP_T [\hat{\eta}_{vt} \hat{\eta}_{v't}] \mid \bU_D(a),\bU_I(a),\bW] - \tE[\tP_T [\eta_{vt} \eta_{v't}] \mid \bU_D(a),\bU_I(a),\bW] \} \\
&= \sqrt{n}\{ \tE[ \hat{g}[\tP_T [\hat{\eta}_{vt} \hat{\eta}_{v't}]] \mid \bU_D(a),\bU_I(a),\bW] - \tE[ \hat{g}[\tP_T [\eta_{vt} \eta_{v't}]] \mid \bU_D(a),\bU_I(a),\bW] \}\\
&= \sqrt{n} \{\tE[\hat{Y}_{j}(a,m) \mid \bU_D(a),\bU_I(a),\bW] - \tE[\Tilde{Y}_{j}(a,m) \mid \bU_D(a),\bU_I(a), \bW]\}\\
& \rightarrow 0,
\end{aligned}
\end{equation*}
as $n,T \rightarrow \infty$. According to Assumption 6, 
\begin{equation*}
\underset{1 \le j \le J}{\max} |\tE[ \Tilde{\Delta}_{jT}(a,a') ]| = o(n^{-1/2}),
\end{equation*}
where $\Tilde{\Delta}_{jT}(a,m) = \tE[\Tilde{Y}_{j}(a,m) \mid \bU_D(a),\bU_I(a),\bW] -  Y_{j}(a,m)$. Therefore, 
\begin{equation*}
\begin{aligned}
&\underset{1 \le j \le J}{\max} |\tE[ \Delta_{jT}(a,a') ]|\\
&= \underset{1 \le j \le J}{\max} |\tE[\hat{Y}_{j}(a,m) \mid \bU_D(a),\bU_I(a),\bW] - \tE[\Tilde{Y}_{j}(a,m) \mid \bU_D(a),\bU_I(a), \bW]|\\
&+ \underset{1 \le j \le J}{\max} |\tE[ \Tilde{\Delta}_{jT}(a,a') ]|\\
&= o(n^{-1/2}),
\end{aligned}
\end{equation*}
as $n,T \rightarrow \infty$. This completes the proof. \qed

\newpage

\textbf{Proof of Theorem 2.} 
We are considering the AIPW estimator $\hat{\psi}_{n,j}^{\text{AIPW}}(a,a')$ in the form that
\begin{equation*}
\begin{aligned}
    \hat{\psi}_{n,j}^{\text{AIPW}}(a,a') &= \tP_n \Big [ \frac{\mathbb{I}(A=a)}{\hat{\pi}_A(a | \bW)} \frac{\hat{\pi}_{A}(a' | M, \bW)\hat{\pi}_A(a | \bW)}{\hat{\pi}_A(a | M, \bW)\hat{\pi}_{A}(a' | \bW)} \left\{\hat{Y}_{j}-\hat{b}_{j}(M,a,\bW) \right\}\\
    &+ \frac{\mathbb{I}(A=a')}{\hat{\pi}_{A}(a' | \bW )}\left\{ \hat{b}_{j}(M,a,\bW ) - \hat{\xi}_{aa'j}^{sr}(\bW) \right\} + \hat{\xi}_{aa'j}^{sr}(\bW) \Big],
\end{aligned}
\end{equation*}
with the influence function
\begin{equation*}
\begin{aligned}
    \phi_{aa'j}(O_i) &= \frac{\mathbb{I}(A_i=a)}{\pi_A(a | \bW_i)} \frac{\pi_{A}(a' | M_i, \bW_i)\pi_A(a | \bW_i)}{\pi_A(a | M_i, \bW_i)\pi_{A}(a' | \bW_i)} \left(\hat{Y}_{ij}-b_j(M_i,a,\bW_i)\right)\\
    &+ \frac{\mathbb{I}(A_i=a')}{\pi_{A}(a' | \bW_i )}\left( b_j\left(M_i,a,\bW_i\right) - \xi_j^{sr}(a,a',\bW_i) \right) + \xi_j^{sr}(a,a',\bW_i) - \psi_j(a,a'),
\end{aligned}
\end{equation*}
which is given by taking the representations of the influence function in the causal mediation literature with Bayes' rule and sequential regression \citep{tchetgen2012semiparametric,vanderweele2015explanation}. Under the existence of an intra-subject processed derived outcome $\hbY$ satisfying Assumptions 6 and 7, inter-subject Assumptions 1-4, and the boundedness conditon in Theorem 2, the AIPW estimator and its influence function is well defined, and $\hat{\psi}_{n,j}^{AIPW}(a,a')$ admits the decomposition as shown in the proof of Equation (8) in Web Supplement Section A.2:
\begin{equation*}
\hat{\psi}_{n,j}^{AIPW}(a,a') - \psi_j(a,a') = \tP_n [\hat{\phi}_j(\tP)] + (\tP_n - \tP)\{\hat{\phi}_j(\hat{\tP}) - \hat{\phi}_j(\tP)\} + \hat{R}_2(\hat{\tP},\tP) + \tE[ \Delta_{jT}(a,a') ].
\end{equation*}
The first term, after scaling by $\sqrt{n}$, converges in distribution to a normal limit by the central limit theorem. The second empirical process term is asymptotically negligible under a Donsker condition or when sample splitting is employed \citep{van2000asymptotic,chernozhukov2018double,kennedy2020sharp}. We can explicitly write the second order remainder term:
\begin{align*}
&\hat{R}_2(\hat{\tP},\tP) \\
&= \hat{\psi}_{n,j}^{AIPW}(a,a') - \hat{\psi}_j(a,a') + \hat{\tP} [\hat{\phi}_j] \\
&= \hat{\psi}_{nj}(a,a') - \hat{\psi}_j(a,a') + \int \hat{\phi}_j(\hat{\tP}) d \tP\\
&= \int \big\{ \frac{I(A=a)}{\hat{\pi}_A(a|\bW)} \frac{\hat{\pi}_A(a'|M,\bW) \hat{\pi}_A(a|\bW)}{\hat{\pi}_A(a|M,\bW) \hat{\pi}_A(a'|\bW)} \left[ \hat{Y}_j - \hat{b}_j(M,a,\bW) \right]\\
&\ \ \ + \frac{I(A=a)}{\hat{\pi}_A(a|\bW)} \left[ \hat{b}_j(M,a,\bW) - \hat{\xi}_{aa'j}(\bW) \right] + \hat{\xi}_{aa'j}(\bW) - \hat{\psi}(\tP) \big\} d \tP\\
&= \int \big\{ \frac{I(A=a)}{\hat{\pi}_A(a|\bW)} \frac{\hat{\pi}_A(a'|M,\bW) \hat{\pi}_A(a|\bW)}{\hat{\pi}_A(a|M,\bW) \hat{\pi}_A(a'|\bW)} \left[ b_j(M,a,\bW) - \hat{b}_j(M,a,\bW) \right]\\
&\ \ \ + \frac{I(A=a)}{\hat{\pi}_A(a|\bW)} \left[ \xi_{aa'j}(\bW) - \hat{\xi}_{aa'j}(\bW) \right] + \hat{\xi}_{aa'j}(\bW) - \hat{\psi}(\tP) \big\} d \tP\\
&= \int \big\{ \frac{I(A=a)}{\hat{\pi}_A(a|\bW)} \left[ \frac{\hat{\pi}_A(a'|M,\bW) \hat{\pi}_A(a|\bW)}{\hat{\pi}_A(a|M,\bW) \hat{\pi}_A(a'|\bW)} - \frac{\pi_A(a'|M,\bW) \pi_A(a|\bW)}{\pi_A(a|M,\bW) \pi_A(a'|\bW)} \right] \left[ b_j(M,a,\bW) - \hat{b}_j(M,a,\bW) \right]\\
&\ \ \ + \frac{I(A=a)}{\hat{\pi}_A(a|\bW)} \frac{\pi_A(a'|M,\bW) \pi_A(a|\bW)}{\pi_A(a|M,\bW) \pi_A(a'|\bW)} \left[ b_j(M,a,\bW) - \hat{b}_j(M,a,\bW) \right]\\
&\ \ \ + \frac{I(A=a')}{\hat{\pi}_A(a'|\bW)} \left[ \xi_{aa'j}(\bW) - \hat{\xi}_{aa'j}(\bW) \right] + \hat{\xi}_{aa'j}(\bW) - \hat{\psi}(\tP) \big\} d \tP\\
&= \int \big\{ \frac{I(A=a)}{\hat{\pi}_A(a|\bW)} \left[ \frac{\hat{\pi}_A(a'|M,\bW) \hat{\pi}_A(a|\bW)}{\hat{\pi}_A(a|M,\bW) \hat{\pi}_A(a'|\bW)} - \frac{\pi_A(a'|M,\bW) \pi_A(a|\bW)}{\pi_A(a|M,\bW) \pi_A(a'|\bW)} \right] \left[ b_j(M,a,\bW) - \hat{b}_j(M,a,\bW) \right]\\
&\ \ \ + \left[\frac{I(A=a)}{\hat{\pi}_A(a|\bW)} - \frac{I(A=a)}{\pi_A(a|\bW)} \right] \frac{\pi_A(a'|M,\bW) \pi_A(a|\bW)}{\pi_A(a|M,\bW) \pi_A(a'|\bW)} \left[ b_j(M,a,\bW) - \hat{b}_j(M,a,\bW) \right]\\
&\ \ \ + \frac{I(A=a)}{\pi_A(a|\bW)} \frac{\pi_A(a'|M,\bW) \pi_A(a|\bW)}{\pi_A(a|M,\bW) \pi_A(a'|\bW)} \left[ b_j(M,a,\bW) - \hat{b}_j(M,a,\bW) \right]\\
&\ \ \ + \frac{I(A=a')}{\hat{\pi}_A(a'|\bW)} \left[ \xi_{aa'j}(\bW) - \hat{\xi}_{aa'j}(\bW) \right] + \hat{\xi}_{aa'j}(\bW) - \hat{\psi}(\tP) \big\} d \tP\\
&= \int \big\{ \frac{I(A=a)}{\hat{\pi}_A(a|\bW)} \left[ \frac{\hat{\pi}_A(a'|M,\bW) \hat{\pi}_A(a|\bW)}{\hat{\pi}_A(a|M,\bW) \hat{\pi}_A(a'|\bW)} - \frac{\pi_A(a'|M,\bW) \pi_A(a|\bW)}{\pi_A(a|M,\bW) \pi_A(a'|\bW)} \right] \left[ b_j(M,a,\bW) - \hat{b}_j(M,a,\bW) \right]\\
&\ \ \ + \left[\frac{I(A=a)}{\hat{\pi}_A(a|\bW)} - \frac{I(A=a)}{\pi_A(a|\bW)} \right] \frac{\pi_A(a'|M,\bW) \pi_A(a|\bW)}{\pi_A(a|M,\bW) \pi_A(a'|\bW)} \left[ b_j(M,a,\bW) - \hat{b}_j(M,a,\bW) \right]\\
&\ \ \ + \left[\frac{I(A=a')}{\hat{\pi}_A(a'|\bW)} - \frac{I(A=a')}{\pi_A(a'|\bW)} \right] \left[ \xi_{a,a'j}(\bW) - \hat{\xi}_{a,a'j}(\bW) \right]\\
&\ \ \ + \frac{I(A=a)}{\pi_A(a|\bW)} \frac{\pi_A(a'|M,\bW) \pi_A(a|\bW)}{\pi_A(a|M,\bW) \pi_A(a'|\bW)} \left[ b_j(M,a,\bW) - \hat{b}_j(M,a,\bW) \right]\\
&\ \ \ + \frac{I(A=a')}{\hat{\pi}_A(a'|\bW)} \left[ \xi_{aa'j}(\bW) - \hat{\xi}_{aa'j}(\bW) \right] + \hat{\xi}_{aa'j}(\bW) - \hat{\psi}(\tP) \big\} d \tP\\
&=\tP \big\{ \frac{I(A=a)}{\hat{\pi}_A(a|\bW)} \left[ \frac{\hat{\pi}_A(a'|M,\bW) \hat{\pi}_A(a|\bW)}{\hat{\pi}_A(a|M,\bW) \hat{\pi}_A(a'|\bW)} - \frac{\pi_A(a'|M,\bW) \pi_A(a|\bW)}{\pi_A(a|M,\bW) \pi_A(a'|\bW)} \right] \left[b_j(M,a,\bW) - \hat{b}_j(M,a,\bW) \right] \\
&+ \left[ \frac{I(A=a)}{\hat{\pi}_A(a|\bW)} - \frac{I(A=a)}{\pi_A(a|\bW)} \right] \frac{\pi_A(a'|M,\bW) \pi_A(a|\bW)}{\pi_A(a|M,\bW) \pi_A(a'|\bW)} \left[ b_j(M,a,\bW) - \hat{b}_j(M,a,\bW) \right]\\
&+ \left[ \frac{I(A=a')}{\hat{\pi}_A(a'|\bW)} - \frac{I(A=a')}{\pi_A(a'|\bW)} \right] \left[ \xi_{aa'j}^{sr}(\bW) - \hat{\xi}_{aa'j}^{sr}(\bW) \right]
\big\}.
\end{align*}
Note that the nuisance rate conditions in Theorem 2 states that we have $\alpha + \beta > \frac{1}{2}$, $\beta + \gamma > \frac{1}{2}$, and $\gamma + \zeta > \frac{1}{2}$, where $\| \pi_{A}(a |m, \bw)\pi_A(a|\bw) - \hat{\pi}_{A}(a | m, \bw)\hat{\pi}_A(a | \bw) \|_2 = \cO(n^{-\alpha})$, $\| b_{j}(m,a,\bw) - \hat{b}_{j}(m,a,\bw) \|_2 = \cO(n^{-\beta})$, $\| \pi_{A}(a|\bw) - \hat{\pi}_A(a | \bw) \|_2 = \cO(n^{-\gamma})$, and $\| \xi_{aa'j}^{sr}(\bw) - \hat{\xi}_{aa'j}^{sr}(\bw) \|_2 = \cO(n^{-\zeta})$. It follows that the second order remainder term $\hat{R}_2(\hat{\tP},\tP)$ is of order $o_{\tP}(n^{-1/2})$. The final bias term $\tE[ \Delta_{jT}(a,a') ]$ in the decomposition is also in the order of $o(n^{-1/2})$ by Lemma 1. Therefore, we have 
\begin{equation*}
( \hat{\psi}_{n,j}^{\text{AIPW}}(a,a') - \psi_j(a,a') ) = \tP_n [\phi_{aa'j}(O)] + o_{\tP}(n^{-1/2}).
\end{equation*}
This completes the proof. \qed

\textbf{Proof of Corollary 1.} The multiple robustness property follows from the structure of the second-order remainder term,
\begin{equation*}
\begin{aligned}
&\hat{R}_2(\hat{\tP},\tP)\\
&=\tP \left\{ \frac{I(A=a)}{\hat{\pi}_A(a|\bW)} \left[ \frac{\hat{\pi}_A(a'|M,\bW) \hat{\pi}_A(a|\bW)}{\hat{\pi}_A(a|M,\bW) \hat{\pi}_A(a'|\bW)} - \frac{\pi_A(a'|M,\bW) \pi_A(a|\bW)}{\pi_A(a|M,\bW) \pi_A(a'|\bW)} \right] \left[b_j(M,a,\bW) - \hat{b}_j(M,a,\bW) \right]\right. \\
&+ \left.\left[ \frac{I(A=a)}{\hat{\pi}_A(a|\bW)} - \frac{I(A=a)}{\pi_A(a|\bW)} \right] \frac{\pi_A(a'|M,\bW) \pi_A(a|\bW)}{\pi_A(a|M,\bW) \pi_A(a'|\bW)} \left[ b_j(M,a,\bW) - \hat{b}_j(M,a,\bW) \right]\right.\\
&+ \left.\left[ \frac{I(A=a')}{\hat{\pi}_A(a'|\bW)} - \frac{I(A=a')}{\pi_A(a'|\bW)} \right] \left[ \xi_{aa'j}^{sr}(\bW) - \hat{\xi}_{aa'j}^{sr}(\bW) \right]
\right\}.
\end{aligned}
\end{equation*}
This remainder term equals zero if any one of the following conditions holds:\\
(i). $\hat{\pi}_{A}(a | m, \bw)$ and $\hat{\pi}_A(a | \bw)$ are consistently estimated;\\
(ii). $\hat{b}_j(m,a,\bw)$ and $\hat{\pi}_A(a | \bw)$ are consistently estimated;\\
(iii). $\hat{b}_j(m,a,\bw)$ and $\hat{\xi}_j^{sr}(a,a',\bW_i)$ are consistently estimated.\\
All remaining terms in the asymptotic expansion are controlled as established in Theorem 2. This completes the proof. \qed

\subsection{Asymptotic normality}

\textbf{Proof of Corollary 2.} 
By Theorem 2, the AIPW estimator admits the asymptotic linear representation
\begin{equation*}
( \hat{\psi}_{nj}^{\text{AIPW}}(a,a') - \psi_j(a,a') ) = \tP_n [\phi_{aa'j}(O)] + o_{\tP}(n^{-1/2}),
\end{equation*}
where $\phi_{aa'j}(O)$ denotes the corresponding influence function. Let $\sigma^2_{aa'j} = \tP [\phi^2_{aa'j}(O)]$. Since $\tP [\phi_{aa'j}(O)]=0$ and $\sigma^2_{aa'j}<\infty$, the central limit theorem yields
\begin{equation*}
    \sqrt{n}\left( \hat{\psi}_{nj}^{\text{AIPW}}(a,a') - \psi_j(a,a') \right) \overset{d}{\rightarrow} N(0, \sigma^2_{aa'j}).
\end{equation*}
This establishes the stated asymptotic normality. \qed

\section{Details of the Multiple Testing Procedure}

\subsection{Simultaneous confidence interval}
\label{sec:simul_CI}

Many biomedical applications involve hundreds of statistical tests. We control the family-wise error rate using Gaussian approximation results developed for the derived outcome framework in \cite{qiu2023unveiling,du2025causal}. To enable valid simultaneous inference, it is necessary to exclude super efficient estimators for which $\theta_j \rightarrow 0$, since in this case the population distribution of the corresponding influence function becomes degenerate \citep{chernozhukov2013gaussian,belloni2018high}.

\begin{condition}\label{condition4_bounded_var}
There exists a non-empty set $\cS^* \subseteq \{1,\cdots,J\}$ of informative estimators such that $\underset{j \in \cS^*}{\min}\, \theta_j^2 > c_5$ for some constant $c_5 > 0$, and $\underset{j \notin \cS^*}{\max}\, \theta_j^2 \rightarrow 0$. In addition, for any distinct $j_1, j_2 \in \cS^*$, the corresponding influence functions satisfy $| \text{cor}(\eta_{j_1}, \eta_{j_2}) | \leq 1 - c_6$ for some constant $c_6 > 0$.
\end{condition}

\cref{condition4_bounded_var} postulates the existence of an informative set $\cS^*$ and imposes boundedness requirements on the variances and covariances of estimators among $\cS^*$. In practice, this condition can be satisfied by screening out estimators $\hat{\tau}_j$ with very small estimated variances and forming $\cS_1 = \{j: \hat{\theta}_j^2 \ge c_7 \}$, where $c_7 > 0$ is a small constant.

For a generic index set $\cS$, we consider the maximal standardized statistic $\kappa_{\cS} = \underset{j \in \cS}{\max} \sqrt{n} |\hat{\tau}_j - \tau_j| \hat{\theta}_j^{-1}$, which is asymptotically equivalent to $\underset{j \in \cS}{\max} \sqrt{n} \theta_j^{-1} | \tP_n [\hat{\eta}_j] |$. Let $\hat{\bm \eta}_{\cS}$ denote the vector collecting $\{ \hat{\eta}_j : j \in \cS \}$, let $\bA_{\cS} = \tP_n [\hat{\bm \eta}_{\cS} \hat{\bm \eta}_{\cS}^T]$ be the corresponding sample covariance matrix, and let $\bD_{\cS}$ be the diagonal matrix with entries $\{ \hat{\theta}_j : j \in \cS \}$. The Gaussian approximation results in the Lemma 11 of \cite{du2025causal} shows that, when $\cS_1 \subseteq \cS^*$, the distribution of $\kappa_{\cS_1}$ can be well approximated by $\| \mathbf{z}_{\cS_1}\|_{\infty}$, where $\mathbf{z}_{\cS_1} \sim N(\bzero, \bD_{\cS_1}^{-1} \bA_{\cS_1} \bD_{\cS_1}^{-1})$. This result enables efficient approximation of the null distribution of $\kappa_{\cS_1}$ using a multiplier bootstrap procedure. Specifically, we generate independent standard normal variables $g_1, \ldots, g_n$, calculate $\mathbf{z}_b = n^{-1/2} \bD_{\cS_1}^{-1} \sum_{i=1}^n g_i \hat{\bm \eta}_{\cS_1,i}$, and repeat this procedure for $b = 1, \ldots, B$. The upper $\alpha$ quantile of $\kappa_{\cS_1}$ is then estimated by $z^{max}_{1-\alpha} = \text{inf} \{ x: \frac{1}{B} \sum_{b=1}^B \mathbb{I}(|\bz_{b}|_{\infty} \leq x) \geq 1 - \alpha \}$. The resulting $(1-\alpha)$ simultaneous confidence intervals for $\{ \tau_j \}_{j\in \cS^*}$ can be constructed as $\hat{\tau}_j \pm z^{max}_{1-\alpha} n^{-1/2} \hat{\theta}_j$.

\subsection{False discovery proportion control}
\label{sec:fdp_control}

The simultaneous confidence intervals may lack power when the number of tests $J$ is large. Therefore, we adopt a multiple testing procedure that controls the false discovery proportion (FDP), defined as the ratio of false discoveries to the total number of discoveries. For a given constant $c > 0$, our goal is to control $\tP(\text{FDP} > c)$, the probability that the FDP exceeds the specified threshold $c$, which is referred to as the FDP exceedance rate (FDPex). In contrast to procedures that control the false discovery rate (FDR), defined as the expectation of the FDP, control of FDPex accounts for the variability of the FDP within a single realized sample and therefore provides a stronger form of error control in high dimensional testing problems.

The FDPex is typically controlled at a prespecified small level $\alpha > 0$ such that $\tP(\text{FDP} > c) \le \alpha$, which yields a stronger form of control on the false discovery proportion and remains asymptotically powerful \citep{genovese2006exceedance,belloni2018high,qiu2023unveiling}. We adopt the same multiple testing procedure developed for doubly robust estimators in \cite{du2025causal}, which extends directly to the multiply robust estimators proposed here for intra-subject processed outcomes.

A step down procedure coupled with a Gaussian multiplier bootstrap is applied to a sequence of nested sets $\cS_1 \supset \cdots \supset \cS_k \supset \cdots \supset \cS_K$ and the associated sequential hypotheses $H^k_0 : \tau_j=0, \ \  \text{vs.}\ \ H^k_1: \tau_j \neq 0, \ \exists \  j \in \cS_k, \  k=1,2,\ldots, K$. We initialize $k = 1$ and set the discovery set $\Omega_1 = \varnothing$. The empirical informative set is defined as $\cS_1 = { j : \hat{\theta}_j^2 \ge c_7 }$ for a small constant $c_7 > 0$. We then use the multiplier bootstrap described in \cref{sec:simul_CI} to approximate the distribution of the maximum standardized statistic $\kappa_{\cS_1}$ over $\cS_1$. If $\kappa_{\cS_1}$ exceeds the upper $\alpha$ quantile $z_{1-\alpha}^{\max}$, we move the index $j_1$ attaining $\kappa_{\cS_1}$ from $\cS_1$ to $\Omega_1$ to obtain updated sets $\cS_2$ and $\Omega_2$. This step down procedure is iterated by testing $\kappa_{\cS_k}$ at each stage and updating the nested sets until the first iteration $k_0$ at which $\kappa_{\cS_{k_0}}$ is not significant. After termination, an augmentation step is applied by adding the indices of the next $\lfloor |\Omega_{k_0}| c/(1-c) \rfloor$ largest values of $|t_j|$ into $\Omega_{k_0}$ to form the final discovery set $\Omega^*$. By Theorem 13 of \cite{du2025causal}, the resulting set $\Omega^*$ satisfies $\tP(\text{FDP} > c) \le \alpha$. The full algorithmic description is provided in Web Supplement Section B.3.

\subsection{The algorithm of the step-down procedure}

\renewcommand{\algorithmicrequire}{\textbf{Inputs:}}
\renewcommand{\algorithmicensure}{\textbf{Output:}}

\begin{algorithm}[H]
\caption{The step-down procedure with augmentation for controlling FDPex}\label{alg:FDP_control}
\begin{algorithmic}[1]
\Require The estimated influence functions $\hat{\eta}_{ij}$ for subjects $i=1,\cdots,n$ and regions $j=1,\cdots,J$. The FDPex threshold $c$ and probability $\alpha$. The threshold for non-informative screening $c_0$. The number of bootstrap samples $B$.
\State Initialize $k=1$. Calculate the estimated variances $\hat{\theta}^2_j = \frac{1}{n} \sum_{i=1}^n \hat{\eta}_{ij}^2$ and the standardized statistics $t_j = \sqrt{n} \hat{\theta}_j^{-1} \sum_{i=1}^n \hat{\eta}_{ij}$ for $j = 1,\cdots,J$. Collect the informative index set $\cS_1=\{j \in 1,\cdots,J : \hat{\theta}_j \geq c_0 \}$ and discovery set $\Omega_1 = \emptyset$. 
\While{not converge}\\
Calculate the maximal statistic over $\cS_k$ as $\kappa_{\cS_k} = \underset{j \in \cS_l}{\text{max}} |t_j|$. \\
Collect $\hat{\bm \eta}_{il}$ as the vectorization of $\{ \hat{\eta}_{ij}: j \in \cS_l \}$, and $\bD_{nl}$ as the diagonal matrix of $\{ \hat{\theta}_j, j \in \cS_l \}$.\\
Generate standard normal variables $g_{1b},...,g_{nb} \overset{iid}{\sim} N(0,1)$, and then calculate $\textbf{z}_{bk} = \bD_{nl}^{-1} \frac{1}{\sqrt{n}} \sum_{i=1}^n g_{ik} \hat{\bm \eta}_{ik}$ for $b=1,...,B$.\\
Approximate the upper $\alpha$ quantile of $\kappa_{\cS_k}$ under $H^k_0$ by $z^{max}_{k,1-\alpha} = \text{inf} \{ x: \frac{1}{B} \sum_{b=1}^B \mathbb{I}(|\bz_{bl}|_{\infty} \leq x) \geq 1 - \alpha \}$

\If{$\kappa_l > z^{max}_{l,1-\alpha}$}
    \State Find $j_k = \underset{j \in \cA_k}{\text{argmax}} |t_j|$, $\cS_{k+1}=\cS_k \setminus \{j_k\}$, and $\Omega_{k+1} = \cS_k \cup \{j_k \}$.
\Else
    \State Stop the step-down procedure.
\EndIf\\
$l = l + 1$
\EndWhile\\
Augment $\cS_k$ with the next $\lfloor |\cS_k|c/(1-c) \rfloor$ largest $|t_j|$ in $\cS_k$, denote as $\cS$.
\Ensure The final discovery set $\cS$.
\end{algorithmic}
\end{algorithm}

\section{Supplementary Materials for the Application to Cortical Surface rs-fMRI Data in ABIDE}

\subsection{Scanner and site information}

Two sites in the ABIDE data contained children ages 8-13: Kennedy Krieger institute and New York University. Resting-state fMRI imaging data from Kennedy Krieger Institute were acquired using one of two protocols: 1) a 3T Philips Achieva scanner, 8-channel head coil, repetition time (TR)/echo time (TE)=2500/30 ms, flip angle 75$^\circ$, 3$\times$3$\times$3 mm voxels, SENSE phase reduction=3, 2 dummy scans, with most scans at either 5 min 20 sec or 6 min 30 sec; or 2) the same protocol except using a 32-channel head coil with most scans 6 min 30 sec. At New York Univerity, MRI scans were acquired using a 3T Siemens Allegra scanner with an 8-channel head coil, TR/TE=2000/15 ms, flip angle=90$^\circ$, 3x3x4 mm voxels, 6 min scan duration. The first two volumes were removed for signal stabilization. An anatomical T1w scan was also collected for each child. 

\subsection{List of stimulant medication}

We define the stimulant medication as in the following list: Dexmethylphenidate, Amphetamine, Dextroamphetamine, Lisdexamfetamine, Methylphenidate, Methylphenidate Extended Release, Methylphenidate Transdermal, Methylphenidate Long Acting.

\subsection{MRI data preprocessing details}

We applied the preprocessing pipeline implemented in \texttt{fMRIPrep} \citep{esteban2019fmriprep} with the \texttt{cifti} option to the T1w and resting-state fMRI data, including anatomical tissue segmentation, surface construction, and surface registration, followed by fMRI motion correction, slice-time correction, boundary-based coregistration, and resampling to the fsaverage template. The detailed output from fMRIPrep is included in Web Supplement Section C.4. We visually inspected the accuracy of the cortical segmentation using the fMRIPrep quality control html files. We excluded participants due to issues with the cortical segmentation. Issues with fMRIprep included image homogeneity issues, outliers in brain morphology, and motion during the T1 scan.

\subsection{Auto-generated text from fMRIprep data preprocessing}

Results included in this manuscript come from preprocessing performed using \texttt{fMRIPrep} 21.0.2 (\citet{fmriprep1}; \citet{fmriprep2}; RRID:SCR\_016216), which is based on \texttt{Nipype} 1.6.1
(\citet{nipype1}; \citet{nipype2}; RRID:SCR\_002502). The text below is automatically produced by \texttt{fMRIprep}.

\textbf{Anatomical data preprocessing} A total of 1 T1-weighted (T1w) images were found within the input BIDS
dataset. The T1-weighted (T1w) image was corrected for intensity non-uniformity (INU) with\texttt{N4BiasFieldCorrection} \citep{n4}, distributed with ANTs 2.3.3 \citep[RRID:SCR\_004757]{ants}, and used as T1w-reference throughout the workflow. The T1w-reference was then skull-stripped with a \texttt{Nipype} implementation of the \texttt{antsBrainExtraction.sh} workflow (from ANTs), using OASIS30ANTs as target template. Brain tissue segmentation of cerebrospinal fluid (CSF), white-matter (WM) and gray-matter (GM) was performed on the brain-extracted T1w using \texttt{fast} \citep[FSL 6.0.5.1:57b01774, RRID:SCR\_002823,][]{fsl_fast}. Brain surfaces were reconstructed using \texttt{recon-all} \citep[FreeSurfer 6.0.1, RRID:SCR\_001847,][]{fs_reconall}, and the brain mask estimated previously was refined with a custom variation of the method to reconcile ANTs-derived and FreeSurfer-derived segmentations of the cortical gray-matter of Mindboggle \citep[RRID:SCR\_002438,][]{mindboggle}. Volume-based spatial normalization to two standard spaces (MNI152 NLin6 Asym, MNI152 NLin 2009c Asym) was performed through nonlinear registration with \texttt{antsRegistration} (ANTs 2.3.3), using brain-extracted versions of both T1w reference and the T1w template. The following templates were selected for spatial normalization: \emph{ICBM 152 Nonlinear Asymmetrical template version 2009c} {[}\citet{mni152nlin2009casym}, RRID:SCR\_008796; TemplateFlow ID: MNI152NLin2009cAsym{]}, \emph{FSL's MNI ICBM 152 non-linear 6th Generation Asymmetric Average Brain Stereotaxic Registration Model} {[}\citet{mni152nlin6asym}, RRID:SCR\_002823; TemplateFlow ID: MNI152NLin6Asym{]}.

\textbf{Functional data preprocessing} For each of the 1 BOLD runs found per subject (across all tasks and sessions), the following preprocessing was performed. First, a reference volume and its skull-stripped version were generated using a custom methodology of \texttt{fMRIPrep}. Head-motion parameters with respect to the BOLD reference (transformation matrices, and six corresponding rotation and translation parameters) are estimated before any spatiotemporal filtering using \texttt{mcflirt} \citep[FSL 6.0.5.1:57b01774,][]{mcflirt}. BOLD runs were slice-time corrected to 1.22s (0.5 of slice acquisition range 0s-2.45s) using \texttt{3dTshift} from AFNI \citep[RRID:SCR\_005927]{afni}. The BOLD time-series (including slice-timing correction when applied) were resampled onto their original, native space by applying the transforms to correct for head-motion. These resampled BOLD time-series will be referred to as \emph{preprocessed BOLD in original space}, or just \emph{preprocessed BOLD}. The BOLD reference was then co-registered to the T1w reference using \texttt{bbregister} (FreeSurfer) which implements boundary-based registration \citep{bbr}. Co-registration was configured with six degrees of freedom. Several confounding time-series were calculated based on the \emph{preprocessed BOLD}: framewise displacement (FD) and three region-wise global signals. FD was computed following Power (absolute sum of relative motions, \citet{power_fd_dvars}) calculated using the implementation in \texttt{Nipype} \citep[following the definitions by][]{power_fd_dvars}. The three global signals were extracted within the CSF, the WM, and the whole-brain masks. The BOLD time-series were resampled into standard space, generating a \emph{preprocessed BOLD run in MNI152NLin2009cAsym space}. First, a reference volume and its skull-stripped version were generated using a custom methodology of \texttt{fMRIPrep}. The BOLD time-series were resampled onto the following surfaces (FreeSurfer reconstruction nomenclature): \texttt{fsaverage}. \texttt{Grayordinates} files \citep{hcppipelines} containing 91k samples were also generated using the highest-resolution \texttt{fsaverage} as intermediate standardized surface space. All resamplings can be performed with \emph{a single interpolation step} by composing all the pertinent transformations (i.e.~head-motion transform matrices and co-registrations to anatomical and output spaces). Gridded (volumetric) resamplings were performed using \texttt{antsApplyTransforms} (ANTs), configured with Lanczos interpolation to minimize the smoothing effects of other kernels \citep{lanczos}. Non-gridded (surface) resamplings were performed using \texttt{mri\_vol2surf} (FreeSurfer).

Many internal operations of \texttt{fMRIPrep} use \texttt{Nilearn} 0.8.1 \citep[RRID: SCR\_001362]{nilearn}, mostly within the functional processing workflow. For more details of the pipeline, see \href{https://fmriprep.readthedocs.io/en/latest/workflows.html}{the section corresponding to workflows in \texttt{fMRIPrep}'s documentation}.

\hypertarget{copyright-waiver}{%
\subsubsection{Copyright Waiver}\label{copyright-waiver}}

The above boilerplate text was automatically generated by \texttt{fMRIPrep} with the express intention that users should copy and paste this text into their manuscripts \emph{unchanged}. It is released under the \href{https://creativecommons.org/publicdomain/zero/1.0/}{CC0} license.

\end{document}